\documentclass[10pt,journal,compsoc]{IEEEtran}

\usepackage{amsmath,algorithm, algorithmicx, algpseudocode,bbm}
\usepackage{textcomp}
\usepackage{amsfonts} 
\usepackage{xcolor}
\usepackage{amsthm}
\usepackage[normalem]{ulem}
\usepackage{subfigure}
\usepackage{graphicx}
\usepackage{cite}
\usepackage{soul}

\usepackage{colortbl} 

\newtheorem{theorem}{Theorem}
\newtheorem{definition}{Definition}
\newtheorem{corollary}[theorem]{Corollary}
\newtheorem{proposition}[theorem]{Proposition}

\definecolor{red}{rgb}{0.0, 0.0, 0.0}
\definecolor{brown}{rgb}{0.0, 0.0, 0.0}
\definecolor{cyan}{rgb}{0.0, 0.0, 0.0}
\definecolor{blue}{rgb}{0.0, 0.0, 0.0}
\definecolor{purple}{rgb}{0.0, 0.0, 0.0}
\definecolor{violet}{rgb}{0.0, 0.0, 0.0}
\definecolor{orange}{rgb}{0.0, 0.0, 0.0}
\definecolor{teal}{rgb}{0.0, 0.0, 0.0}

\definecolor{magenta}{rgb}{0.0, 0.0, 0}
\definecolor{green}{rgb}{0.0, 0.0, 0}

\ifCLASSINFOpdf
\else
\fi
\begin{document}
%
\title{FIRE: A Failure-Adaptive RL Framework for Edge Computing Migrations}

\author{\IEEEauthorblockN{Marie~Siew, Shikhar~Sharma, Zekai~Li, Kun~Guo, \textit{Member, IEEE}, Chao~Xu, \textit{Member, IEEE},\\ Tania~Lorido-Botran, Tony~Q.S.~Quek, \textit{Fellow, IEEE}, Carlee~Joe-Wong, \textit{Senior Member, IEEE}}
\thanks{Marie Siew and Tony Q.S Quek are with the Information Systems and Technology and Design Pillar, Singapore University of Technology and Design, Singapore 487372.}
\thanks{Shikhar Sharma, Zekai Li, and Carlee Joe-Wong are with the Electrical and Computer Engineering department, Carnegie Mellon University, 15213 Pittsburgh, PA, USA.}
\thanks{Kun Guo is with the Shanghai Key Laboratory of Multidimensional Information Processing, School of Communications and Electronics Engineering, East China Normal University, Shanghai 200241, China.}
\thanks{Chao Xu is with the School of Information Engineering, Northwest A\&F University, Yangling 712100, China.}
\thanks{Tania Lorido-Botran is with Roblox, USA.}
}

\maketitle

\begin{abstract}
In edge computing, users’ service profiles are migrated between edge servers due
to user mobility. Reinforcement Learning (RL) frameworks have been
proposed to do so, often trained on simulated data. However, existing
\textcolor{magenta}{RL} frameworks overlook occasional server failures, which although rare,
impact latency-sensitive applications like AR/VR and real-
time obstacle detection. These rare failures, being
not adequately represented in historical training data, pose a challenge
for data-driven \textcolor{magenta}{RL} algorithms. 
We introduce FIRE, a
framework that adapts to rare events by training a \textcolor{magenta}{RL} policy in an edge
computing digital twin environment. We propose FIRE-ImRE, an importance sampling-based Q-learning algorithm, which samples rare events proportionally to their impact on the value function. FIRE considers
delay, migration, failure, and backup placement costs across individual
and shared service profiles. We prove FIRE-ImRE’s boundedness and
convergence to optimality. Next, we introduce novel deep Q-learning
(FIRE-ImDQL) and actor critic (FIRE-ImACRE) versions of our algorithm
to enhance scalability. We extend our framework to accommodate users
with varying risk tolerances of rare failure events. 
Through trace-driven experiments, we show
that FIRE reduces edge computing costs compared to vanilla \textcolor{magenta}{RL} and the greedy baseline
in the event of failures.

\end{abstract}

%
\IEEEpeerreviewmaketitle

\section{Introduction}
\label{section:Introduction}

Mobile (or multi-access) edge computing (MEC) is a key technology in 5G and 6G telecommunication systems. It enables computationally intensive and latency sensitive mobile applications such as real time image processing, augmented reality \textcolor{magenta}{(AR)}, 
interactive gaming, etc. 
In MEC, computing resources such as server clusters are positioned close to 
end-users, typically
at cellular base stations or WiFi Access Points \textcolor{magenta}{(AP)} at the edge of the radio access network \cite{mao2017survey, PavelMECsurvey}. 
This proximity reduces response latency by avoiding the wide-area network delays encountered in cloud computing \cite{PavelMECsurvey}.

Researchers have been actively investigating how to jointly allocate computing and radio resources to computing tasks offloaded to edge servers \cite{mao2017survey, Xuchen16}. 
Nevertheless, 
user mobility across geographical areas presents a challenge towards task offloading \cite{SWangSM_MDP, OuyangFollowMe,wang2018survey}. 
Service migration, or moving service profiles to \textcolor{magenta}{APs} nearer the user as they move, has been proposed to maintain service continuity \cite{OuyangFollowMe,wang2018survey}.
Yet, frequent service migration 
causes additional expenditure, due to the required network usage. 
Therefore, balancing the delay-migration cost tradeoff has been widely studied \cite{SWangSM_MDP, wang2019delay, OuyangFollowMe, networkAccessSP, MSiew2021VM}. 
\textcolor{purple}{Due to the large decision space 
when the numbers of \textcolor{magenta}{APs} and time-slots increase, and the lack of information on costs and mobility patterns,
Reinforcement Learning (RL) is a common solution method \cite{SWangSM_MDP,wang2019delay, RLnetworkEdge}. \textcolor{cyan}{Such methods often rely on digital twins of the edge computing network to train \textcolor{magenta}{RL} models, as digital twins} allow the exploration of different resource allocation policies, without generating real-world consequences.
}

\subsection{Challenges: Edge Computing and Resilience}
\textcolor{blue}{
Few migration works, however, account for another challenge in edge computing: the resilience of edge computing systems to \textit{rare but serious events like server failures}.
}
They occur due to reasons such as system overload, hardware failures, or malicious attacks, and they can be costly:
the average cost of an outage at a cloud data center, for example, has increased to $\$740000$ in $2016$ \cite{ponemon2016cost}. 
Edge server failures may be even more probable than those in the cloud due to their distributed geographical locations, which complicates management and maintenance~\cite{aral2020learning}\textcolor{cyan}{, though they will likely remain rare overall}.
\textcolor{cyan}{Even worse, edge computing failures can} have a sizeable impact as many edge computing applications are latency sensitive, such as the Internet of vehicles, \textcolor{magenta}{AR}, and video processing applications \cite{mao2017survey}.
\textcolor{blue}{
Managing resources without considering rare but severe failure increases latency and reduces reliability.}
\textcolor{brown}{For instance, if the user's service profile is migrated to a server that then fails, its job is not able to be completed,} jeopardizing the \textcolor{brown}{smooth and safe functioning of applications, such as autonomous driving and real time obstacle detection.} 


\textcolor{brown}{We propose to use \emph{backups},} 
which can take over when the primary service's edge server fails, to handle failures. 
Nevertheless, managing the migration of both the primary and backup services is challenging, \textcolor{red}{as the number of \textcolor{magenta}{APs} and timeslots increases.}
As mentioned above, due to the large decision space, and the network operator's lack of information, \textcolor{magenta}{RL} has traditionally been proposed to solve the service migration problem \cite{SWangSM_MDP,wang2019delay, RLnetworkEdge, 8705822}.
Nevertheless, rare events and failures occur at a low probability, making it difficult to jointly plan or learn an optimal resource allocation policy for both the usual and rare event scenarios.
\textcolor{brown}{This is because \textit{RL's reliance on past reward data may overlook the significance of rare events, impacting the training process 
\cite{frank2008reinforcement}.}}
Therefore, we introduce \textbf{FIRE}: Failure-adaptive Importance sampling for Rare Events, a \textit{resilience} framework for edge computing service migration.






\subsection{FIRE: Failure-adaptive Importance sampling for Rare Events}
\textcolor{violet}{Our FIRE framework employs a digital twin - inspired setup, employing \textcolor{magenta}{RL} to optimize service migration and the \textit{backup migration} in edge computing systems, amidst potential rare service failures.}
\textcolor{violet}{
A digital twin edge computing system virtually replicates 
a physical edge computing system and mirrors its network conditions, akin to 
similar setups used for data center management and remote monitoring \cite{dataCenterDigitalTwin}. 
Our framework FIRE \textcolor{cyan}{further models limited edge server capacity and communication delays, }
and \textcolor{purple}{our optimization balances the trade-offs amongst the following costs:} \emph{communication and computing delay costs} of the multiple users, the \emph{migration cost} of service profiles, the \emph{backup storage and migration costs}, along with the \emph{failure costs}.}

Because the actual rare event probability is low and unable to help us learn an optimal policy to prepare for rare events, \textcolor{violet}{we use \textbf{importance sampling of rare events} in our digital twin setup to increase the sampling of the rare events in the RL environment.}
Error correction is done through the use of importance sampling weights, to enable us to learn an optimal policy with respect to the true rare events probability.
\textcolor{violet}{Unlike online training, training of RL in the edge computing digital twin also enables the system to learn policies while 
\emph{avoiding the large failure costs} of experiencing actual server failures.} 
The converged policy trained by the simulator is then applied to online scenarios, in which rare events happen at their natural probabilities.
\textcolor{teal}{While \cite{precup2000eligibility,frank2008reinforcement} also use importance sampling to sample rare events, these two works involve estimating the value function given a fixed policy, and their algorithm is not applied in the edge computing scenario. 
In contrast, FIRE learns, not just evaluates, the optimal policy $\pi^*$. We further propose novel boundedness and convergence proofs. \emph{To the best of our knowledge, we are the first to propose importance sampling of rare events to prepare for failures in edge computing.}}

\textcolor{violet}{We model two decision making scenarios. Firstly, we suppose that every user has an individual service profile for job computation. Secondly, we consider multiple users sharing service profiles, for example common game environments when users are playing the same mobile or AR game together, or common neural networks for multi-modal learning
\textcolor{purple}{\cite{zhang2019dynamic}.}} 
\textcolor{blue}{The service profile replicas can also serve as backups in this second scenario\textcolor{cyan}{~\cite{aral2018replica}, while in the first scenario we must create an individual backup for each user}.} 
\textcolor{cyan}{We further consider that edge computing users may} exhibit varying risk tolerances due to differences in their computing applications and latency requirements, or different willingness to incur the backup placement or migration costs.

In summary, our \textbf{contributions} are as follows: 
\begin{itemize}
    \item 
    We present a \textcolor{cyan}{\textbf{model} of service migration and backup placements given user mobility} 
    that introduces server failures as rare events. \textcolor{purple}{We model both the individual and shared service profile settings.} 
    \item We introduce \textcolor{brown}{\textbf{FIRE}}, a \textit{resilience} framework designed to handle rare events like server failures. 
    \textcolor{cyan}{As a digital twin, o}ur importance sampling based Q-learning algorithm \textbf{FIRE-ImRE} 
    simulates rare events frequently without real-world consequences\textcolor{cyan}{, allowing us to learn a failure-aware migration policy. 
    }
    We prove that this algorithm is bounded (Theorem 2), and that it converges to the optimal policy (Theorem 4). 
    \item We propose 
    \textcolor{violet}{deep Q-learning (\textbf{FIRE-ImDQL}) and actor critic (\textbf{FIRE-ImACRE}) versions} of FIRE, to handle large and combinatorial state and action spaces in real-world networks. \textcolor{red}{These algorithms differ from traditional deep Q-Learning and actor critic as we incorporate importance sampling and error correcting weights.}
    Additionally, \textcolor{cyan}{we present the \textbf{RiTA} algorithm 
    that adjusts service migration and backup placement based on individual risk tolerance, without separate RL training for different risk tolerances.}
    \item Finally, we provide \textbf{trace-driven simulation results} showing the convergence of our algorithms to optimality. We show that unlike vanilla Q-learning, FIRE's variants are resilient towards server failures, resulting in sizeable cost reductions.
    \textcolor{purple}{Our importance sampling framework can be extended to solve other resource allocation and decision making problems in light of rare events in networking scenarios.}
\end{itemize}

The paper is structured as follows. Section \ref{section:related} discusses the related work. \textcolor{magenta}{Sections \ref{section:SysModel} and \ref{section:SysModel2} present the system models, on the individual and shared service profile settings respectively.} Section \ref{section:OptProb} presents the problem formulation. Following which, we present our solution and its proofs in Section \ref{section:AlgoSolution}, and the deep Q-learning and Actor Critic algorithms in Section \ref{section:FnApproxDQN}. Next, 
Section \ref{section:riskLevel} addresses the heterogeneous risk tolerance scenario.
Finally, we present simulation results in Section \ref{section:Simul} and conclude in Section \ref{section:conclusion}.





\section{Background and Related Work}
Building on the literature on service and virtual machine (VM) migration in follow-me-cloud, 
\textbf{service or VM placement and migration} in light of user mobility in edge computing has been a widely studied problem \cite{wang2018survey,SWangSM_MDP, RLnetworkEdge, wang2019delay, OuyangFollowMe, networkAccessSP, MSiew2021VM, kim2022modems}. 
These are often NP-hard knapsack-like problems.
In particular, \cite{SWangSM_MDP, wang2019delay} used dynamic programming and RL to make migration decisions given unknown mobility information. 
\cite{OuyangFollowMe} optimized the long term performance-cost trade-off 
via Lyapunov optimization.
Nevertheless, these service migration works have not addressed resilience, including preparation for and adaptation to server failures. 

Other works have realized that \textbf{cloud and edge computing presents resilience challenges} to resource allocation, including 
ultra-reliable offloading mechanisms using two-tier optimization~\cite{liu2019dynamic}, 
proactive failure mitigation for edge-based NFV~\cite{huang2019proactive}, failure aware workflow scheduling for cloud computing \cite{cloudWorkflowSchedule}, RL for replica placement in edge computing~\cite{liang2021two}. 
and graphical models to learn spatio-temporal dependencies between edge server and link failures \cite{ponemon2016cost}. The above works did not consider resilience in light of user mobility.
The migration of users at overlapping coverage areas has been considered in \cite{du2021optimal}, but here, users at the center of coverage areas still experience a higher latency \textcolor{cyan}{as they are} connected to the cloud.
Nevertheless, the above works did not consider 
unknown user mobility and unknown costs. \textcolor{orange}{Our prior work \cite{siew2023acre} 
addressed rare events in single-user service migration using an actor-critic importance sampling algorithm. Here, we extend this to multi-user scenarios with both individual and shared service profiles, introduce tabular and deep Q-learning algorithms, and provide theoretical guarantees for the tabular algorithm.
}

\textcolor{magenta}{We do not assume that failures can be predicted, unlike most other failure mitigation works. Instead, we aim to learn a policy that builds in resiliency to rare events that are not predicted in advance. Note that our algorithm can complement strategies like checkpointing \cite{mudassar2022adaptive} or dynamically adjusting the number of service backups. It addresses the unique challenge of resiliency in reinforcement learning-based policies, which fundamentally suffer from a lack of rare events in their training datasets.}

Other edge computing challenges have been explored from the \textbf{risk-awareness} perspective, catering to heterogeneous users and applications, for example risk-aware non-cooperative MEC offloading \cite{riskAwareOffloading1}, 
and energy budget risk-aware application placement \cite{riskAwareApplicationPlc}. \textcolor{cyan}{However, they do not consider the effect of rare events on learned RL policies.}

\textbf{(Deep) RL} has been used to optimize migration and offloading in edge computing \cite{SWangSM_MDP, riskAwareOffloading1,wang2019delay, RLnetworkEdge, 8705822} 
with the goal of optimizing across trade-offs such as migration cost, latency and energy costs. 
Importance sampling has been used in RL for policy evaluation (i.e., finding the expected future reward from following a given policy $\pi$) \cite{precup2000eligibility, frank2008reinforcement}. 
However, these two works involve estimating the value function given a fixed policy, instead of learning the optimal policy $\pi^*$ over time. Their proof techniques do not generalize to the $Q$-learning 
~\cite{sutton2018reinforcement} \textcolor{cyan}{that forms the core of FIRE}.
Recent work has proposed using importance sampling to enhance RL's data efficiency~\cite{madhushani2020hamiltonian} in general, but they do not consider rare events explicitly and 
require knowledge of the full system dynamics. 
\label{section:related}

\section{Migration of Individual Service Profiles} 
\label{section:SysModel}
\textcolor{purple}{In this section, we introduce our failure-aware service migration model, with individual service profiles. }
\textcolor{cyan}{In the next section, we will consider shared service profiles.}
\textcolor{magenta}{Our system models, which capture rare events, can be integrated into a ``digital twin'' for offline algorithm learning in light of rare events such as server failures.  The purpose of digital twins is to allow the system to learn a policy in a simulated and virtual environment, without actually experiencing the cost of rare events.} \textcolor{green}{For example, a simulated edge computing environment may serve as a digital twin \cite{dong2019deep,zhang2021adaptive}.} 

\subsection{Service Placement Model}\label{subsection:SvcPlacementModel}
\textbf{Users and edge servers.} We consider an MEC system with \textcolor{violet}{$K$ users and} $N$ \textcolor{magenta}{APs} (APs), such as a base station or wifi-access point. Each AP is equipped with a server to which users can offload latency-sensitive computing jobs\textcolor{cyan}{, and we associate each AP with a disjoint physical region; users in this region then connect to this \textcolor{magenta}{AP}}. 
Users are mobile and move from region to region at different times of the day. For example, office workers may move to the city for work during morning rush hours, changing their nearest AP.
User $k$'s location at time $t$ is represented by \textcolor{violet}{$l_u^k(t)$}.
We consider that each user will be associated with its nearest AP, via local area network (LAN).
The users' mobility pattern follows a probability distribution $m(l_{u}'^{k}|l_u^{k})$, where $m(l_{u}'^{k}|l_u^{k})$ represents the transition probability the user travels to location $l_{u}'^{k}$, after being at location $l_u^{k}$. 

\textcolor{cyan}{Each user maintains a \textit{service profile} at the edge servers, e.g., in \textcolor{magenta}{VMs} or containers \cite{OuyangFollowMe}, that consists of the state needed to run its edge computing application.} 
\textcolor{cyan}{Our goal is then to decide where the service profile and any backups should be located.}
\textcolor{black}{To maintain low latencies for time-sensitive service applications, in light of user mobility, services are pre-migrated - likely to nearby APs. }
\textcolor{purple}{Service migration is stateful, involving the migration of the set of parameters and intermediate results associated with the user's computing job.} \textcolor{black}{At each time-slot $t$, the user sends a service request to the network operator.} 
The location of the user's service at time $t$ is represented by $l_s^k(t)$, \textcolor{brown}{and due to pre-migrations the network operator makes this placement decision before the user location at time $t$ is known.} 

\textbf{Failure model.} In reality, rare events (system anomalies) such as server failures or shutdowns occur. 
Such events will impact the user's quality of service, e.g., if the user is not able to have his or her job computed on time, impacting the safe and smooth functioning of edge applications and jeopardizing the low latencies of edge computing. \textcolor{green}{We define a failure as any event that prevents a user's service job from being completed, which may stem from either a hardware or a software failure.}

We consider two server types in the network. The first type models larger edge datacenters, which tend to have more safety mechanisms like fully duplicated electrical lines with transfer switches that help servers come back online fast after a failure, causing just a slight delay to users. 
The second type models a smaller edge data center that is less equipped with mechanisms to help deal with failure.
These servers take longer to come back online, resulting in a longer delay and higher probability of jobs not being served at all, hence a higher cost.
\textcolor{purple}{We let $f^k_{ind}(t)$ represent the failure indicator, which indicates whether $l_s^k(t)$, the location of user $k$'s service profile, experiences failure at time $t$.
}
\textcolor{orange}{Failures are Markovian, and under Server Type two, there is a higher probability of failures continuing into the next time-slot.}

\begin{figure}[t]
\centering
\includegraphics[
angle=0,scale=0.255]{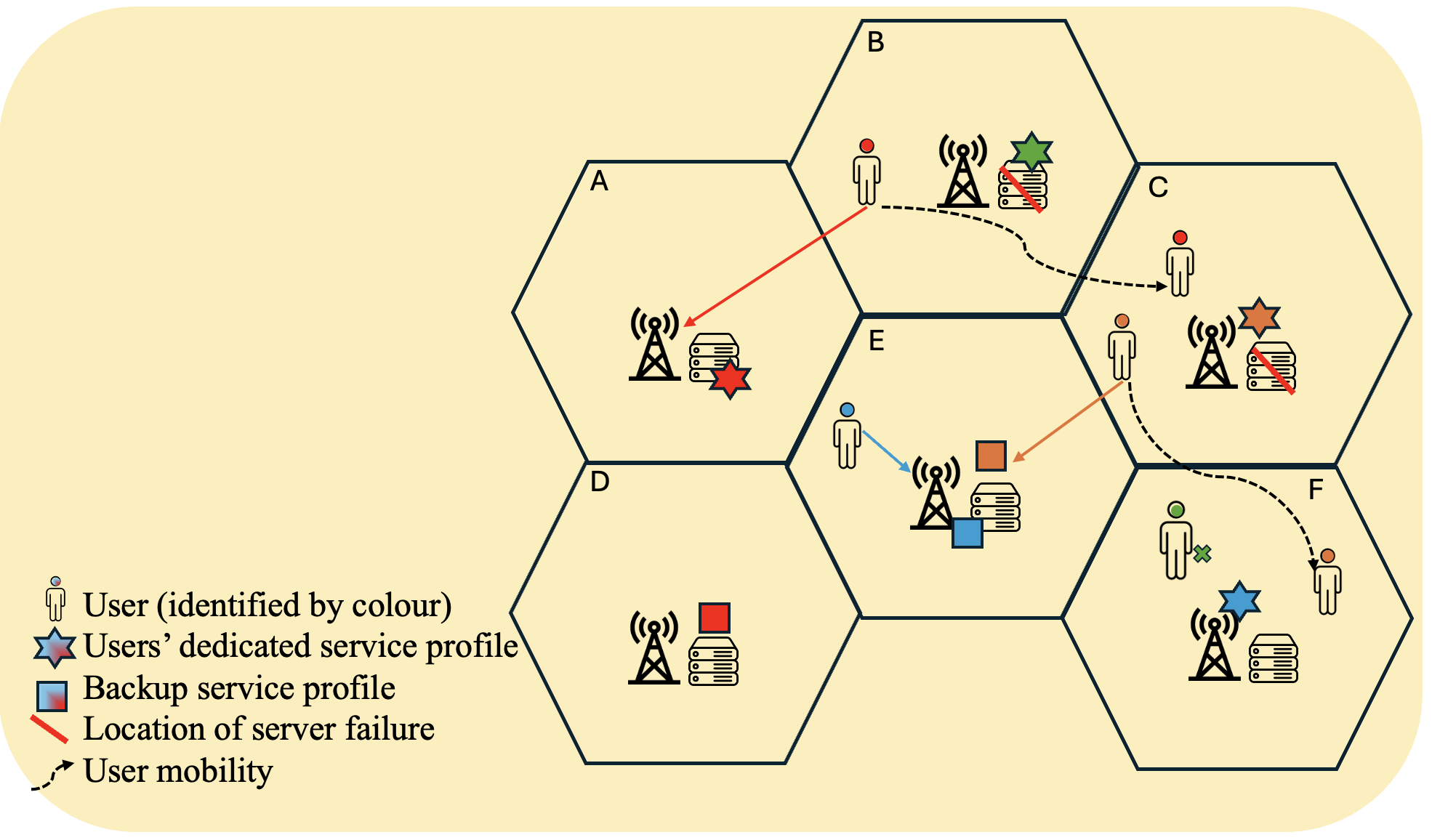}
\caption{\textcolor{magenta}{Users move across the network. Each user has its individual service profile, which has a backup, in case the server where their service profile is at fails.} 
\textcolor{magenta}{To evaluate system resilience and optimize failure-handling algorithms and strategies without incurring real-world costs, the setup is modeled within a digital twin.}
}
\label{fig:sysModel}
\end{figure}

To take into account the potential failures and reduce the costs experienced by the user, a \textbf{backup} of the user's service profile can be placed in the system.
This way, if the server $l_s^k(t)$ at which user $k$'s service is placed at experiences a failure, the user will still be able to offload its job.
\textcolor{brown}{As seen in Fig. \ref{fig:sysModel}, unlike the green user which does not have a backup service profile, the brown user is still able to offload its job computation even though the server at AP C failed, because there is a backup at location E. 
}

We formulate a \textbf{Markov Decision Process} (MDP) with system state $s(t) \in S$ is \textcolor{violet}{
$[s_1(t), .. s_k(t),... s_K(t)]$, where 
\begin{equation}
s_k(t)=(l_u^k(t), l_s^k(t), f_{ind}^k(t), b_{ind}^k(t)).
\label{eq:state}
\end{equation}
} 
$l_u^k$ is user $k$'s current location, $l_s^k$ is the location of user k's service, $f_{ind}^k$ is an indicator on whether the user's service is placed at a location with a server failure, and $b_{ind}^k$ indicates whether there currently is a backup for user $k$'s service in the system.
At each timeslot $t$, the network operator takes action \textcolor{violet}{$a(t)=[a_1(t), ... a_k(t), ... a_K(t)]$, where 
\begin{equation}
a_k(t)=( l_s^k(t+1), b_u^k(t+1) )
\end{equation}}
\textcolor{brown}{pre-emptively, before the users moves to their next locations $[l_u^k(t+1)]$.} This involves determining the placement location of the service $l_s^k(t+1)$, and the position of the backup in the system $b_u^k(t+1) \in \{0,1,... N, \text{no backup} \}$. 
These decisions will directly impact the system state $s$ at the next timeslot.

\textcolor{cyan}{Transitions between states are governed by users' movements between APs ($l_u^k(t)$ and $l_u^k(t+1)$) following their mobility pattern $m(l_u'|l_u)$.
To quantify them, we let}
\begin{equation}
    \textcolor{violet}{h(\tilde{s}|s,a)=Pr(a) \prod^K m(l_u'|l_u),}
    \label{eq:intermediateTransProb}
\end{equation} 
the product of the user's mobility distribution and the probability that the network operator takes a particular action,
where \textcolor{violet}{$\tilde{s}(t)
=[\tilde{s}_1(t), .. \tilde{s}_k(t),... \tilde{s}_K(t)]$} and
\begin{equation}
    \tilde{s}_k(t+1)=(l_u^k(t+1), l_s^k(t+1), b_{ind}^k(t+1)) \in \tilde{S}   
    \label{eq:intermediateState}
\end{equation} 
refers to the state record without the failure indicator $f_{ind}$. We let the set of these records be $\tilde{S}$.
At each state $s$, there is a small probability $\epsilon (s)$ (e.g., $\epsilon (s)=0.01$), that in the next timeslot, a server failure occurs \textcolor{violet}{in at least one location}. \textcolor{orange}{The failure indicator $f_{ind}$ is then set to $1$. 
These states belong to the set of \textit{``Rare event states"}, $T \subset S$,} \textcolor{orange}{defined as follows:}
\begin{definition}
A subset of states $T \subset S$ is called the Rare Event State Set if the following properties hold:


1. There exists $s \in S, s' \in T$ for which $p(s'|s,a)>0$.

2. Let $T^{\pi}(s)$ denote the contribution of the rare states towards the value of state $s$, according to Eq. (\ref{eq:T_eqnDef}). 
For a given policy $\pi$, there exists $s \in S$ for which
\begin{equation}
\label{eq:AlgoAssumptionV}
    |T^{\pi}(s)| \gg 0
\end{equation}

3. For every $s' \in T$, $f_{ind}=1$, indicating that the location of at least one user's service profile $l_s^k$ is experiencing a failure.
\end{definition}
Property 1 means that transition to the rare event state set $T$ is possible.
Property 2 
means that the rare event states (set $T$) collectively have a sizeable (non-negligible) impact on the value function and hence on the system's cost. 

With a larger probability of $1-\epsilon(s)$, no server failure occurs, i.e. $f_{ind}=0$. The non rare event states \textit{``Normal States"} are defined as the set $S \backslash T$. Hence we have
\begin{equation}
\begin{aligned}
    s_k(t+1) &=\tilde{s}_k(t+1) \cup f_{ind}^k(t+1)\\
   & =(l_u^k(t+1), l_s^k(t+1), b_{ind}^k(t+1), f_{ind}^k(t+1) )
\end{aligned}
\end{equation}
Therefore, the overall state transition probability of the system can be expressed as
\begin{equation}
    p(s'|s,a)= \begin{cases} (1-\epsilon(s))h(\tilde{s}|s,a) , & \text{if} \ s' \notin T \\
    \epsilon(s) h(\tilde{s}|s,a), & \text{if} \ s' \in T,
    \end{cases}
    \label{eq:OverallTransProb_Rewritten}
\end{equation}

\subsection{Reward Function}

We next specify the reward function of our MDP. We aim to learn a policy that minimizes the total cost of migration, including operational, failure and backup costs.
\subsubsection{\textcolor{purple}{Operational Costs}} 
\textcolor{black}{The network operator migrates the users' profiles across \textcolor{magenta}{APs}, 
to deliver a better QoS (lower job delay) for the user.} 
\textcolor{cyan}{We first model the costs incurred in normal (non-failure) states, which include communication and computing delays, both of which affect user QoS, as well as the cost of migrating service profiles.}
We let $d^{\text{comm},k}_{l_u,j}$ denote the communication delay the user faces when it offloads its job to the edge server, given that its current location is $l_u^k$ and its service is placed at \textcolor{magenta}{AP} $j$. 
$d^{\text{comm},k}_{l_u,j}$ consists of the access latency of uploading the job to the user's associated \textcolor{magenta}{AP}, and the transfer latency of forwarding the job to the edge server's location if the service is placed at another \textcolor{magenta}{AP} (i.e. $l_u^k \neq j$). This transferring latency via LAN depends on the hop count of the shortest communication path \cite{ouyang2019adaptive}.
Therefore, the delay is a function of the distance between $l_u^k$ and $j$.
Therefore, the \textcolor{violet}{\textbf{total communication delay cost of $K$ users }
at time $t$ is}

\begin{equation}
    D(t)=\sum_{k=1}^K \sum_j d^{\text{comm}}_{l_u^k,j} \mathbbm{1}_{\{l_s^k(t)=j\}},
    \label{eq:CommDelay1User}
\end{equation}
where $\mathbbm{1}_{\{l_s^k=j\}}$ is the indicator function, which will take the value of $1$ if the user's service profile is placed at \textcolor{magenta}{AP} $j$, and will take the value of $0$ otherwise.

\textcolor{violet}{As multiple users share the resources at an edge node, users will experience a compute delay. Modelling the process at each AP as an M/M/1 queue, 
the \textbf{total computing delay experienced across all APs} at time $t$ is \cite{networkAccessSP}: 
}
\begin{equation}
    C(t)=\sum_{j \in A}\sum_{k=1}^K \mathbbm{1}_{\{l_s^k(t)=j\}}  \frac{1}{x_j - \sum_{k=1}^K z_k \mathbbm{1}_{\{l_s^k(t)=j\}} }
    \label{eq:compdelay}
\end{equation}
\textcolor{violet}{where $x_j$ refers to the capacity at AP $j$, $z_k$ is user $k$'s task size, and $\sum_{k=1}^K z_k(t) \mathbbm{1}_{\{l_s^k(t)=j\}}$ is the total load at AP $j$ at time $t$. 
The total computing delay is the sum over all locations $j$ and users $k$, of the average queuing delay each user experiences.}
\textcolor{purple}{As users share computing resources at an AP, the presence of other users' jobs will impact the queuing (i.e. computing) delay experienced by a user. 
}

The migration cost $m_{i,j}$, which includes the operational and energy costs on network devices like routers and switches \cite{OuyangFollowMe}, 
is a function of the distance across two APs $i$ and $j$. 
The \textbf{total migration cost of $K$ users} at time $t$ is 

\begin{equation}
    M(t)=\sum_{k=1}^K \sum_{i \in A} \sum_{j \in A} m_{ij} \mathbbm{1}_{\{l_s^k(t-1)=i\}}\mathbbm{1}_{\{l_s^k(t)=j\}},
    \label{eq:switchingCost}
\end{equation}
where $\mathbbm{1}_{\{l_s^k(t-1)=i\}}$ is the indicator function that equals $1$ if the service profile was placed at location (AP) $i$ at time $t-1$.

\subsubsection{Failure and Backup Costs}
\textcolor{magenta}{We next specify the failure costs and the costs incurred by provisioning backups in case of failures. As mentioned earlier, the system can be modeled within a digital twin setup, allowing algorithm testing and policy learning in a virtual and simulated setting - to avoid experiencing the real-world cost.}
There is a \textbf{backup cost} $B(t)$ incurred by \emph{storing} the backup at a server: 
a sum of the backup migration cost and the backup storage cost,
as it takes up space and prevents other content from being stored.
\begin{equation}
\begin{aligned}
    \textcolor{brown}{B}(t)&=\sum_{k=1}^K \sum_{j\in A} \rho_j \mathbbm{1}_{\{b_u^k(t)=j\}}\\
    &+\sum_{k=1}^K \sum_{i \in A} \sum_{j \in A} m_{ij} \mathbbm{1}_{\{b_u^k(t-1)=i\}}\mathbbm{1}_{\{b_u^k(t)=j\}},
    \label{eq:backupCost}
\end{aligned}
\end{equation}
where $\rho_j$ is the storage cost at \textcolor{magenta}{AP} $j$.
There is also a \textbf{failure cost} $F(t)$ incurred
\begin{equation}
    F(t)= \sum_{k=1}^K F \mathbbm{1}_{\{f_{ind}^k(t)=1\}} (\mathbbm{1}_{\{b_u^k(t)=\text{no backup}\}} \lor \mathbbm{1}_{\{b_u^k(t)=l_s^k(t)\}})
    \label{eq:failureCost}
\end{equation}
\textcolor{violet}{This cost is experienced at time $t$ when i) location $l_s^k(t)$ has a failure at time $t$ $(f^k_{ind}(t)=1)$,
and ii) there is no backup in the system ($b_u^k(t)=\text{no backup}$),} or when the placement of the backup $b_u^k(t) = l_s^k(t)$, the location of the main service.
Otherwise, if there is a backup placed, there will be no failure cost. The presence of failures themselves, i.e., $f_{ind}(t)$, is however independent of the user actions. 



\section{Migration of Shared Service Profiles} 
\label{section:SysModel2}
\begin{figure}[t]
\centering
\includegraphics[
angle=0,scale=0.24]{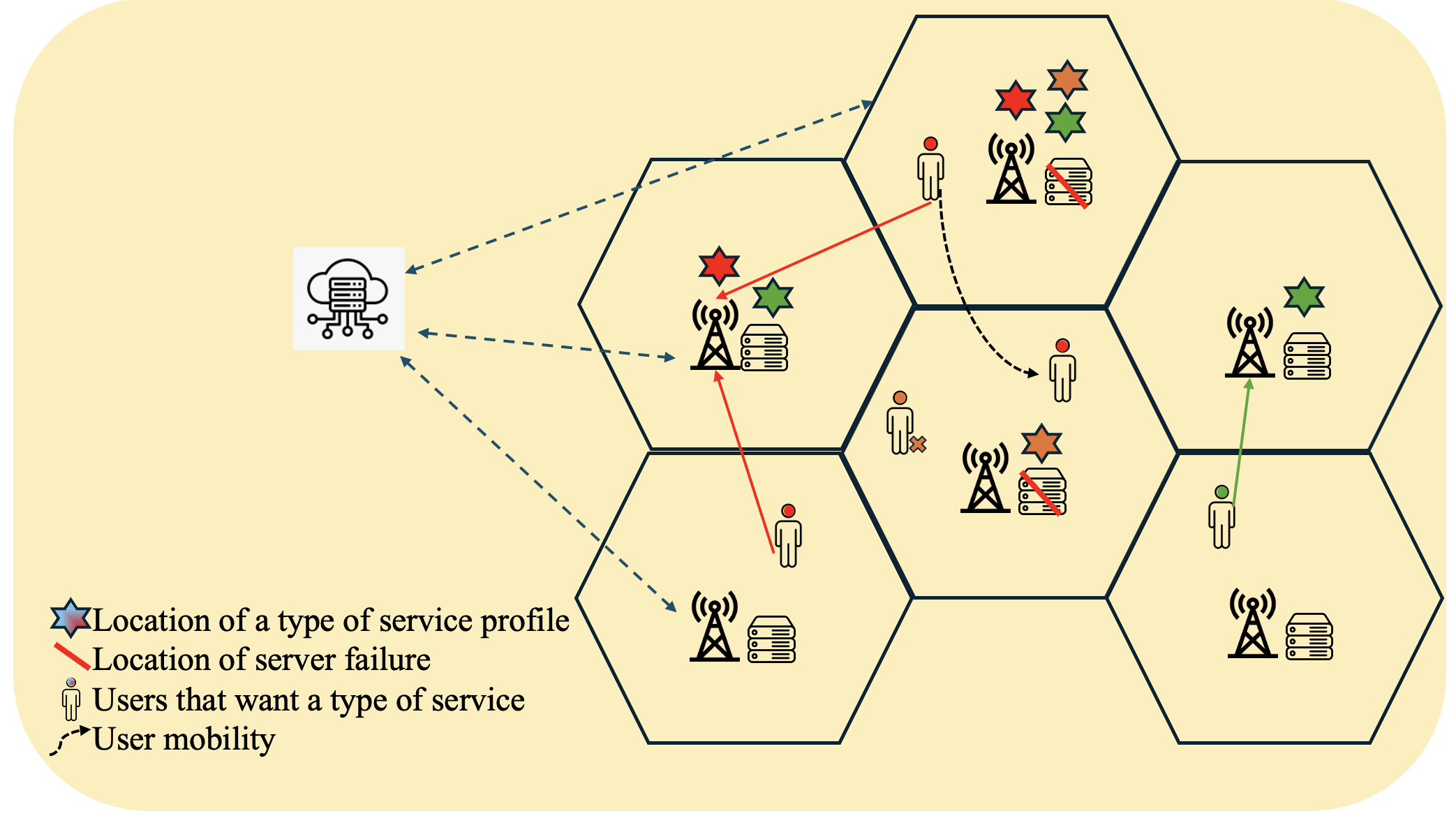}
\caption{\textcolor{magenta}{Multiple users (who are mobile across the network) share a service profile (SP) for their task, e.g. a common game environment or neural network. For instance, the users in red share the SP in red, which has multiple copies in the network.}}
\label{fig:SysModel2}
\end{figure}
\textcolor{magenta}{In this section, we introduce the edge computing service migration scenario where, differing from Section \ref{section:SysModel},} multiple users share a common service profile (SP) for the same task. This could be a common neural network for multi-modal learning, or a common game environment for their \textcolor{magenta}{AR} or mobile gaming applications \textcolor{purple}{such as when multiple users are playing a game together \cite{zhang2019dynamic}}. 
Therefore, the network operator is concerned with the placement of shared SPs across the network, in light of the geographical demand distributions of the different SPs.
In the event of a server failure, the impact will be multiplied when SPs are shared, as more users cannot have their job served.

\textcolor{cyan}{As in the case of individual service profiles, we} formulate a \textbf{Markov Decision Process} (MDP) in which the system state is
$[s_{1,1}(t), .. s_{i,j}(t),... s_{I,J}(t)]$, where 
for each location, SP pair ($i,j$), their component of the state space at time $t$ is
\begin{equation}
    s_{i,j}(t)=(u_{i,j}(t), d_{i,j}(t-1), f_i(t), t_i(t)),
    \label{eq:SharedSP_StateSpace}
\end{equation}
where $u_{i,j}(t)$ is the number of users who want to use SP $j$ at location $i$, $d_{i,j}(t-1)$ indicates whether SP $j$ has been previously deployed at location $i$, and $f_i(t)$ and $t_i(t)$ denote the failure status and the server type at location $i$ respectively.
\textcolor{purple}{Here, there are multiple identical copies of profiles being placed, as many users share them (see Fig. \ref{fig:SysModel2}).}
\textcolor{purple}{Note that this state space (concatenation over all location - SP pairs) is different from that in the individual service profile scenario in Section \ref{section:SysModel} (concatenation over individual user information).
Here the placement decisions depend on the location and mobility distributions of the users who \textit{share} the profile, whereas in Section \ref{section:SysModel} the placement of every profile is dependent on the location of the corresponding individual user.}
The \textcolor{cyan}{corresponding} action space will be
\begin{equation}
        [ (l^{SP1}_{\text{loc}1}, l^{SP1}_{\text{loc}2}, ..., l^{SP1}_{\text{loc}N}), ...(l^{SPj}_{\text{loc}1}, l^{SPj}_{\text{loc}2}, ...), ...],
        \label{eq:SharedSP_ActSpace}
    \end{equation}
where $l^{SP j}_{\text{loc}\ i} \in \{0,1\}$ is a binary variable indicating whether we place or migrate SP $j$ at location $i$.

\textcolor{cyan}{As in the case of individual service profiles, we aim to learn a policy that minimizes the total costs, which include both operational and backup and failure-related costs.} 

The \textcolor{purple}{\textbf{total communication delay} cost across all SP types $j \in SP$ and locations $i \in A$,
at time $t$ is}
\begin{equation}
    D^{sh}(t)=\sum_{j \in SP} \sum_{i \in A} u_{ij}(t) \min_{ \hat{i} \in \{ \hat{i} | l^{SP j}_{\text{loc}\ \hat{i}}=1\}} d^{\text{comm}}_{i,\hat{i} },
    \label{eq:CommDelay_Shared}
\end{equation}
\textcolor{purple}{where $d^{comm}_{i, \hat{i}}$ is the communication delay from location $i$ to $\hat{i}$, same as Section \ref{section:SysModel}. The users at location $i$ access the nearest AP which contains the SP type they require ($\min_{ \hat{i} \in \{ \hat{i} | l^{SP j}_{\text{loc}\ \hat{i}}=1\}} d^{\text{comm}}_{i,\hat{i} }$).
\textcolor{cyan}{QoS is also affected by the} \textbf{total computing delay experienced across all APs} at time $t$ \cite{networkAccessSP}:}
\begin{equation}
    C^{sh}(t)=\sum_{i \in A} \frac{1}{x_i - \sum_{j \in SP} \sum_{\hat{i} \in A} load_{j,\hat{i}}(t)}
    \label{eq:CompDelay_Shared}
\end{equation}
\textcolor{purple}{where $x_i$ is the capacity of the servers at AP $i$. The load at each location $i$ is a sum over all SP types $j$, all other locations $\hat{i}$ (inclusive of $i$). The load $load_{j,\hat{i}}(t)=z_j u_{i,j}(t) 
\wedge \mathbbm{1}_{\{i =\text{argmin}_{\{m\ | l^{SP j}_{loc\ m}=1\} } d^{comm}_{m,\hat{i}}  \}} $, the product of task $j$'s size $z_j$, the number of users at location $\hat{i}$ who need SP type $j$, $u_{i,j}(t)$, 
and the indicator function indicating that location $i$ is the nearest location containing SP $j$, to the users at $\hat{i}$. \textcolor{cyan}{As in the case of individual service profiles, service profiles at the same AP must share resources, leading to increased delays.}
}

The \textbf{total migration cost} at time $t$ is
\begin{equation}
\begin{aligned}
    & M^{sh}(t) = \sum_{j \in SP} \mathbbm{1}_{\{d_{i,j}(t-1)=0\ \&\ d_{i,j}(t)=1\}} \times \big[ g_{cloud,i} \\ \label{eq:MigrationCost_Shared}
    & \mathbbm{1}_{\{ d_{\hat{i},j}(t-1) =0,\ \forall \hat{i}\}} + \min_{\{\hat{i}| d_{\hat{i},j}(t-1)=1 \}} m_{\hat{i},i} \mathbbm{1}_{\{\exists \hat{i}\ s.t.\ d_{\hat{i},j}(t)=1\}}
    \big]. \nonumber
\end{aligned}
\end{equation}
\textcolor{purple}{
The migration cost for SP $j$ kicks in if it is to be placed at location $i$ at time $t$ and was not placed there at $t-1$ 
(i.e. $d_{i,j}(t-1)=0\ \&\ d_{i,j}(t)=1)$. It consists of either downloading SP $j$ from the cloud server for a new deployment when $d_{\hat{i},j}(t-1) =0,\ \forall \hat{i}$, or migrating SP $j$ from the nearest location $\hat{i}$ when $\exists \hat{i}\ s.t.\ d_{\hat{i},j}(t)=1$. 
The \textbf{total storage cost} of storing all the deployed SPs is}
\begin{equation}
    S^{sh}(t)= \sum_{j \in SP} \sum_{i \in A} h_j \rho_i \mathbbm{1}_{\{l_{\text{loc}\ i}^{SP j}(t)=1\}},
    \label{eq:storage_Shared}
\end{equation}
\textcolor{purple}{where $h_j$ is the size of service profile type $i$ and $\rho_i$ is the per unit storage cost at AP $i$. The \textbf{failure cost} is incurred for the users who need to use service profile type $j$ when either a) SP type $j$ is not placed anywhere ($l_{\text{loc}\ i}^{SP j}=0\ \forall i$), b) the delay exceeds $del^j$,  the delay threshold for this application ($\min_{ \hat{i} \in \{ \hat{i} | l^{SP j}_{\text{loc}\ \hat{i}}=1\}} d^{\text{comm}}_{i,\hat{i} } > del^j\}$), or c) there is a failure at a location $\hat{i}$, and there is no available SP $j$ placed elsewhere, or the minimum delay elsewhere exceeds the delay threshold $del^j$ for this application.}
\begin{equation}
\begin{aligned}
    &F^{sh}(t) =\sum_{ j \in SP, i \in A} u_{ij} [ \mathbbm{1}_{\{l_{\text{loc}\ i}^{SP j}=0\ \forall i\}} \\
    & \vee \mathbbm{1}_{\{ \min_{ \hat{i} \in \{ \hat{i} | l^{SP j}_{\text{loc}\ \hat{i}}=1\}} d^{\text{comm}}_{i,\hat{i} } > del^j\}} \notag \vee \Pi_{\hat{i}} ( \mathbbm{1}_{\{ f_i(t)=1 \}} \\ 
    & \wedge ( \mathbbm{1}_{\{l_{\text{loc}\ i}^{SP j}=0\ \forall \hat{i} \neq i\}} \notag \vee \mathbbm{1}_{\{ \min_{ \hat{i} \in \{ \hat{i} \neq i | l^{SP j}_{\text{loc}\ \hat{i}}=1\}} d^{\text{comm}}_{i,\hat{i} } > del^j\}} ))] \notag
    \label{eq:failureCost_shared}
\end{aligned}
\end{equation}

\section{Cost Minimization in Light of Rare Events}
\label{section:OptProb}
\textcolor{purple}{In this section, we introduce the service migration optimization problems, for both the individually used (Section \ref{section:SysModel}) and shared (Section \ref{section:SysModel2}) service profile scenarios.}
Our focus is on long term service placements and backup storage decisions, taking into account potential rare events such as server shutdowns \textcolor{magenta}{($f^k_{ind}(t)=1$)} which have an impact on the user experience.
\textcolor{purple}{For both problems,} we aim to minimize the expected sum
of the 
\textbf{\textcolor{violet}{weighted sum of the communication and computing delay}, migration, storage and failure costs}.
\textcolor{blue}{Actions are taken in a digital twin framework, as they allow for algorithm training while avoiding the cost of failures in real-world applications. 
}
\textcolor{purple}{Firstly, we introduce the optimization problem for the \textcolor{magenta}{scenario in Section \ref{section:SysModel}} where all users have their individual service profiles.} 

\begin{equation}
\begin{aligned}
\min_{\pi(s)}\ &
\mathbb{E}_{\pi}[\sum_t (w_D D(t)+ w_C C(t) + w_M M(t)+w_B B(t)\\ 
&+w_F F(t))] \label{eq:cost}
\end{aligned}
\end{equation}
\begin{equation}
\begin{aligned}    
\text{s.t.}\ & [s_1(t), .. s_k(t),... s_K(t)]
\sim p(s'|s,a).\\ 
& \text{Eqs.} (\ref{eq:CommDelay1User}), (\ref{eq:switchingCost}), 
(\ref{eq:compdelay}),(\ref{eq:backupCost}),
(\ref{eq:failureCost}),
\end{aligned}
\end{equation}
\textcolor{purple}{where $w_D, w_C,$ etc. are the corresponding weights, which can be adjusted based on the priority of different costs.}
The decision variable $\pi(s,a)=\{P_r(a_t=a|s_t=s)\}$ is the network operator's policy, the probability it will take each action $a(t)$ (service placement and backup storage decision for all users) given the state $s(t)$. 
The first constraint indicates that the state dynamics follows the transition probability of the system in Eq. (\ref{eq:OverallTransProb_Rewritten}), incorporating rare event transitions, while the other constraints indicate the delay, migration, backup and failure costs respectively.

Next, we introduce the optimization problem for the \textcolor{magenta}{scenario in Section \ref{section:SysModel2} where users share service profiles}; \textcolor{cyan}{note that we aim to minimize the same types of costs as in Eq.~\eqref{eq:cost} with individual service profiles}.
\begin{equation}\label{eq:OF2}
\begin{aligned}
\min_{\pi(s)}\ &
\mathbb{E}_{\pi}[\sum_t (w_D D^{sh}(t)+ w_C C^{sh}(t) + w_M M^{sh}(t)\\ & + w_S S^{sh}(t)
+F^{sh}(t))]
\end{aligned}
\end{equation}
\begin{equation}
\begin{aligned}
\text{s.t.}\ & [s_{1,1}(t), .. s_{i,j}(t),... s_{I,J}(t)] 
\sim p(s'|s,a).\\ 
& \text{Eqs.} (\ref{eq:CommDelay_Shared}), (\ref{eq:MigrationCost_Shared}), (\ref{eq:CompDelay_Shared}),
(\ref{eq:storage_Shared}),
(\ref{eq:failureCost_shared}).
\label{eq:constraint-ProbDist}
\end{aligned}
\end{equation}
The decision variable $\pi(s,a)$ is the network operator's policy, the probability it will take each action $a(t)$ (placement decisions of the service profiles ) given the current state $s(t)$. 
The first constraint indicates that the state dynamics follows the system's transition probability, which incorporates rare event transitions, while the other constraints indicate the delay, migration, storage and failure costs.
\textcolor{red}{
\begin{proposition}
The number of migration path possibilities grows at least fast as $O(N^T)$.
\end{proposition}
\begin{proof}
For each timeslot $t$, there are $N$ possible locations to migrate a service profile. Since the service profile must be placed at a server each timeslot, the result follows.
\end{proof}
Deciding a migration path for the user across $T$ timeslots is similar to the shortest path problem. 
However, the network operator may not have information on mobility distributions and the corresponding latency costs, complicating look-ahead solutions. 
Furthermore, the costs (e.g. computing delays, Eqs. (\ref{eq:compdelay}) and (\ref{eq:CompDelay_Shared})) of different users are coupled through service profiles' sharing of AP resources.  
Thus, Markov Decision Processes and RL have been used to derive service migration policies \cite{wang2019delay, RLnetworkEdge, 8705822}. 
The difficulty in our problem further arises as we model occasional server failures. 
As these failures occur at a low probability, and because RL relies on past reward data, the learner might overlook rare events and fail to learn an optimal failure-adaptive policy. 
We introduce our solution in the next section.
} 

\section{RL in the presence of Server Failures}
\label{section:AlgoSolution}
In this section, we present FIRE: Failure-adaptive Importance sampling for Rare Events, a RL based framework in an edge computing digital twin. 
\textcolor{magenta}{Our algorithm is applied to solve both optimization problems introduced in Section \ref{section:OptProb}.}
\textcolor{teal}{While \cite{precup2000eligibility} and \cite{frank2008reinforcement} also use importance sampling to sample rare events, these two works involve estimating the value function given a fixed policy. Their proof techniques also do not generalize to the $Q$-learning. In contrast, our algorithms involve learning the optimal policy $\pi^*$, which involves policy changes. Based on our algorithm, we propose novel boundedness and convergence proofs.}
We 
present our algorithm (FIRE-ImRE) and its proof in Subsections \ref{subsection:AlgoDescription} and \ref{subsection:algoProof} respectively, and present an eligibility traces version (ETAA) for faster convergence in \textcolor{cyan}{the Appendix}.
We propose a Deep Q learning and Actor Critic version of our algorithm in Section \ref{section:FnApproxDQN} \textcolor{cyan}{to accommodate large state and action spaces.} 
\textcolor{teal}{FIRE is compatible with any RL formulation, and our algorithms are examples of this.}

\subsection{\textcolor{brown}{Rare Events and the Value Function}} 
\label{subsection:DefineRareEv}
\textcolor{magenta}{We define rare event states as occasions when service fails: $f^k_{ind}(t)=1$ in Section \ref{section:SysModel} or $f_i(t)=1$ in Section \ref{section:SysModel2}.}
We are concerned with rare event states if they have a sizeable impact on the user experienced costs, 
i.e. if these states collectively have an impact on the value function. 
The value function $V^{\pi}(s): S \rightarrow \mathbb{R}$ for the policy $\pi$ is the expected-return of state $s$, i.e. it indicates how ``good" (cost wise) a state is (based on the discounted sum of future rewards), when $\pi$ is used. The discount factor $\gamma$ indicates how much the network operator cares about future vs current rewards. 
\begin{equation*}
\label{eq:value_def}
    V^{\pi}(s)=\mathbb{E}_{a_t,\space  p(s'|s,a)}\left[\sum_{t=1}^{\infty} \gamma^k r( s(t), a(t), a(t+1))| s_0=s\right]
\end{equation*}
where the reward (cost) is the per-timeslot cost  
in (\ref{eq:cost}) or (\ref{eq:OF2}). 
The value function is the solution to the recursive Bellman equation \cite{sutton2018reinforcement}:
\begin{equation*}
    V^{\pi}(s) = \sum_{a \in A} \pi(s,a) \sum_{s' \in S} p(s'|s,a)[r(s,a,s')+\gamma V^{\pi}(s')].
\end{equation*}
Defining $T^{\pi}(s)$ as the collective contribution of the rare states towards $V^{\pi}(s)$ \cite{frank2008reinforcement} given a fixed policy $\pi$, \textcolor{cyan}{where the rare states $T$ are as defined in Section~\ref{subsection:SvcPlacementModel},} we have
\begin{equation}
    T^{\pi}(s)=\epsilon(s) \sum_{a \in A} \sum_{ \tilde{s}' \in \tilde{S} } h(\tilde{s}'|s,a)\pi(a|s)[ r(s,a,s')+ \gamma V^{\pi}(s')] 
    \label{eq:T_eqnDef}
\end{equation}
Likewise, defining $U^{\pi}(s)$ as the collective contribution of the non-rare states (states in $S\backslash T$) towards the value of state $s$, given a fixed policy $\pi$, we have :
\begin{equation}
\begin{aligned}
    &U^{\pi}(s)=(1-\epsilon(s) )\times\\
    &\sum_{a \in A} \sum_{\tilde{s}' \in \tilde{S}} \pi(a|s) h(\tilde{s}'|s,a)[r(s,a,s')+ \gamma V^{\pi}(s')].
    \label{eq:U_eqnDef}
    \end{aligned}
\end{equation}

The Q-value of a state action pair $(s,a)$ is the expected-return of taking action $a$ at state $s$, following $\pi$. 
\begin{equation*}
    Q^{\pi}(s,a)=\mathbb{E}_{\pi}[\sum_{k=1}^{\infty} \gamma^k r( s(t), a(t), s(t+1))| s_t=s, a_t=a]
\end{equation*}
\begin{equation*}
   = \sum_{s'\in S}p(s'|s,a)[r(s,a,s')+\gamma \sum_{a'\in A} \pi(a'|s') Q^{\pi}(s',a')]
\end{equation*}
The optimal policy will achieve the minimum cost and minimum value function:
\begin{equation}
    Q^*(s,a)=\min_{\pi} Q^{\pi}(s,a).
\end{equation}

\subsection{FIRE-ImRE: Q-Learning in Light of Rare Events}
\label{subsection:AlgoDescription}

In this subsection, we present our importance sampling based failure-adaptive algorithm.
\textcolor{red}{We use a digital twin-based simulator in order to avoid the large failure costs experienced during server failures in online scenarios,}
\textcolor{violet}{and the converged policy trained in the digital twin is then applied to online scenarios, in which rare events happen at their natural probabilities.}
Our algorithm performs importance sampling while increasing the rare event probabilities from $\epsilon(s)$ to $\hat{\epsilon}(s)$. The transition probabilities $h(\tilde{s}|s,a)$, which mostly depend on the user mobility probabilities (Eq. (\ref{eq:intermediateTransProb})), are obtained from historical data. 
We use importance sampling because if rare events such as server failures are sampled at their natural probabilities, \textcolor{brown}{the RL algorithm is not able to converge to the optimal policy\textcolor{cyan}{; in particular,} it may decide that the storage cost isn't worth having backups. This can be catastrophic in the event of failures. In our numerical simulations in Section \ref{section:Simul}, we show that unlike our algorithm, a traditional RL algorithm is unable to train an optimal policy which avoids experiencing high cost in the event of failures (Figs. \ref{fig:compAlgo_AC} - \ref{fig:sharedSP_results}).}

\textcolor{brown}{\textbf{Idea behind algorithm:}} 
\textcolor{orange}{Firstly, we introduce importance sampling: when we want to calculate the expectation of a function $f(x)$, and the true distribution $p$ is difficult to sample from, importance sampling involves sampling from another distribution $q$, to help compute $E[f(x)]$ \cite{bucklew2004introduction}. Here, the expectation of $f(x)$ can be estimated as }
\begin{align}
E[f(x)] & =\int f(x) p(x) d x \\
&=\int f(x) \frac{p(x)}{q(x)} q(x) d x \\ & \approx \frac{1}{n} \sum_{i} f\left(x_{i}\right) \frac{p\left(x_{i}\right)}{q\left(x_{i}\right)}.
\end{align}
\textcolor{orange}{When $q$ is well crafted, variance in the estimates will be reduced. 
For our service migration and backup placement problem, we will sample the rare events at probability $\hat{\epsilon}(s)$ (analogous to $q(x)$), a higher rate than $\epsilon(s)$ (analogous to $p(x)$), which is the probability of visiting a rare event state from state $s$.}
\textcolor{brown}{Therefore, to obtain $\hat{\epsilon}(s)$,} we define $T(s,a)$ as the contribution of the rare states $s' \in T$ towards the Q-value of the state action pair $(s,a)$, and $U(s,a)$ as the contribution of the normal states $s' \in S \backslash T$ towards the Q-value of the state action pair $(s,a)$:
\begin{equation}
    T(s,a)=\epsilon(s) \sum_{\tilde{s}'\in \tilde{S}} h(\tilde{s}'|s,a)[r(s,a,s')+\gamma \max_b Q(s',b)]
\end{equation}
\begin{equation}
    U(s,a)=(1-\epsilon(s)) \sum_{\tilde{s}'\in \tilde{S}} h(\tilde{s}'|s,a)[r(s,a,s')+\gamma \max_b Q(s',b)]
\end{equation}
where $s'= \tilde{s}' \cup f_{ind}$. 
With this, we have the relationship $T(s,a)+U(s,a)=Q(s,a).$
A potential update equation 
would be
$T(s,a) \gets (1-\alpha_T) T(s,a)+\alpha_T \epsilon(s) (r(s,a,s') +\gamma \max_b Q(s',b))$, where $\alpha_T$ is the learning rate.

Updating at the state action $(s,a)$ level may lead to slower convergence, \textcolor{brown}{especially with an increasing number of \textcolor{magenta}{APs}. 
For example, with 9 \textcolor{magenta}{APs}, there will be 360 states and 90 actions, leading to 32400 state-action combinations.}
We can reduce this complexity by exploiting our earlier observation that \emph{the probability of rare events (i.e., server failures) happening is independent of the actions chosen}. 
Therefore, we propose updating $T$ and $U$ (the contributions of the rare events and non-rare events respectively towards the Q-value) 
at the state level:
\begin{equation*}
\begin{split}
    &T(s)=\\
    &\epsilon(s) \sum_{a \in A} \sum_{\tilde{s}'\in S} \pi (a|s) h(\tilde{s}'|s,a) [r(s,a,s')+\gamma \max_b Q(s',b)]
    \end{split}
\end{equation*}
\begin{align*}
    U(s)&=(1-\epsilon(s)) \times \sum_{a \in A} \sum_{\tilde{s}'\in S} \pi (a|s) h(\tilde{s}'|s,a) \big[r(s,a,s') \\
    & +\gamma \max_b Q(s',b)\big]
\end{align*}
These equations contain a sum over the actions.
The relationship $U(s)+T(s)=\sum_a Q(s,a) \pi(a|s)$ holds. 
To help us obtain the importance sampling rare event rate $\hat{\epsilon}(s)$,
We will use the following respective update equations for $T(s)$ and $U(s)$ in our algorithm:
\begin{equation*}
T(s) \gets (1-\alpha_T) T(s) + \alpha_T \epsilon(s) (r(s,a,s')+\gamma \max_b \hat{Q}(s',b))
\end{equation*}
\begin{equation*}
\begin{aligned}
    U(s) \gets &(1-\alpha_U) U(s)\\
    & + \alpha_U (1-\epsilon(s)) (r(s,a,s')+\gamma \max_b \hat{Q}(s',b)).
\end{aligned}
\end{equation*}

\textcolor{brown}{\textbf{Importance sampling and correction:}}  Based on the above,
we will calculate the rare event importance sampling rate $\hat{\epsilon}(s)$ as follows:
\begin{equation}
\hat{\epsilon}(s) \gets \min(\max(\delta, \frac{|T(s)|}{ |T(s)|+|U(s)|} ) ,1-\delta).
\end{equation}
The bounds $(\delta,1-\delta)$ ensure sufficient rare event sampling. 
As importance sampling of rare events takes place according to $\hat{\epsilon}(s)$, 
$\hat{p}(s'|s,a)$ is the transition probability in our algorithm which incorporates importance sampling. \textcolor{orange}{It follows Eq. (\ref{eq:intermediateState}), with $\epsilon(s)$ being replaced by $\hat{\epsilon}(s)$.} 



Because the rare events are sampled at the probability $\hat{\epsilon}$ instead of their actual probability $\epsilon$, we need a method of correction, in order to learn the optimal policy for the original system with transition probability $p(s'|s,a)$ (Eq. (\ref{eq:OverallTransProb_Rewritten})). 
Importance sampling correction weights \textcolor{orange}{$w(s,a,s')$ will be used}, when obtaining the temporal-difference (TD) error in Q-learning.
They are obtained through:
\begin{equation}
    w(s,a,s') \gets \begin{cases} \epsilon(s)/ \hat{\epsilon}(s), & \text{if} \ s' \in T,\\
    (1-\epsilon(s))/(1-\hat{\epsilon}(s)), & \text{if} \ s' \notin T,
    \end{cases}
    \label{eq:ImpWeight}
\end{equation}
where $1-\epsilon(s)$ is the probability that $s'$ is a non-rare event state.
The TD error will be 
\begin{equation}
    w_t (r(s,a,s') + \gamma \max_b\hat{Q}(s',b)) - \hat{Q}(s,a),
\end{equation}
instead of the traditional Q-learning TD error $ r(s,a,s')+ \gamma \max_b \hat{Q}(s',b) - \hat{Q}(s,a)$.

\textbf{Algorithm description:} The algorithm is presented in Algorithm \ref{algo:combined} (\textbf{FIRE-ImRE}). 
Firstly, we initialise $\hat{Q}(s,a)$, $\hat{T}(s)$ and $\hat{U}(s)=0$, and $\hat{\epsilon}(s)=\frac{1}{2}$, for all states, and initialise the learning rates $\alpha^t,\alpha^t_T, \alpha^t_U$.
For every timeslot, 
the importance sampling rare event probability $\hat{\epsilon}(s)$ will determine whether or not a rare-event (failure) occurs. Thereafter, the rest of the state transition will occur according to the probability distribution $h(\tilde{s}|s,a)$, where $\tilde{s}(t)=(l_u(t), l_s(t), b_{ind}(t))$.
This would result in a new state $s^{t+1}$, and a reward value $r^{t+1}=r(s^t,a^t,s^{t+1})$ (line 4).
Based on the next state $s^{t+1}$, a new action is selected according to the $\beta$-greedy policy (lines 5-6), which is the same as the traditional $\epsilon$-greedy policy in Q-learning, where with probability $\beta$ the greedy action is selected and with probability $1-\beta$ a random action is selected, for exploration. We call it $\beta$-greedy to avoid confusion with our rare event probability $\epsilon$.

Next, we obtain the importance weight $w_t$ for this timeslot. The importance weight is the actual probability divided by the importance sampling - based probability (line 7).
The temporal-difference (TD) error $\triangle_t$ will be updated (line 8), while involving error correction using the importance sampling weight $w_t$. 
With the TD-error, we update $\hat{Q}(s^t,a^t)$, the Q-value for state $s^t$ and action $a^t$ (line 9). The size of update is determined by the learning rate $\alpha^t$.
We then update either $T(s^t)$ or $U(s^t)$, depending on whether the next state $s^{t+1}$ is a rare event state or not (lines 10-14).
Finally, the importance sampling rare events probability $\hat{\epsilon}(s^t)$ is updated in line 15, according to $\frac{|T(s^t)|}{ |T(s^t)|+|U(s^t)|}$, and bounded by $\delta$ and $1-\delta$.
The process iterates for every timeslot, until we have achieved convergence.
\textcolor{violet}{The converged policy will be applied online.}

\begin{algorithm}
\caption{FIRE-ImRE: Importance Sampling Q-Learning for Rare Events}\label{algo:combined}
\begin{algorithmic}[1]
\State \textbf{Initialise:} $\hat{Q}(s,a)$ randomly, $\hat{T}(s), \hat{U}(s) \gets 0$, $\hat{\epsilon}(s) \gets \frac{1}{2}, \alpha^t,\alpha^t_T, \alpha^t_U.$ 
\State Select the initial state $s^0$ and action $a^0$.
\For{all timeslots $t$}
\State $\hat{\epsilon}(s^t)$ determines if an rare event happens. Thereafter, sample according to $h(\tilde{s}|s,a)$. The new state $s^{t+1}$ and a reward value $r^{t+1}=r(s^t,a^t,s^{t+1})$ is observed. 
\State $a^{t+1} \gets \beta-\text{greedy}(s^{t+1}).$ 
\State $a^* \gets \text{argmax}_b \hat{Q}(s^{t+1},b).$ 
\State $w_t \gets \begin{cases} \epsilon(s^t)/ \hat{\epsilon}(s^t), & \text{if} \ s^{t+1} \in T.\\
    (1-\epsilon(s^t))/(1-\hat{\epsilon}(s^t)), & \text{if} \ s^{t+1} \notin T.
    \end{cases}$
\State $\triangle_t \gets w_t (r^{t+1} + \gamma \hat{Q}(s^{t+1},a^*)) - \hat{Q}(s^{t},a^t)$
\State $\hat{Q}(s^t,a^t) \gets \hat{Q}(s,a) + \alpha^t \triangle_t $.
\If{$s^{t+1} \in T$}
\State $T(s^t) \gets (1-\alpha^t_T) T(s^t) + \alpha^t_T \epsilon(s^t) (r^{t+1}+\gamma \hat{Q}(s^{t+1},a^*))$
\Else
\State $U(s^t) \gets (1-\alpha^t_U) U(s^t) + \alpha^t_U (1-\epsilon(s^t)) (r^{t+1}+\gamma \hat{Q}(s^{t+1},a^*))$
\EndIf
\State $\hat{\epsilon}(s^t) \gets min(max(\delta, \frac{|T(s^t)|}{ |T(s^t)|+|U(s^t)|} ) ,1-\delta)$
\EndFor

\end{algorithmic}
\end{algorithm}

Our algorithm works for the scenario where $\exists s, \text{s.t.} |T(s)| \gg 0$.
See Statement 2 in the definition of rare events in Definition 1.
This condition means that the rare events have a sizeable contribution to the Q-value of at least one other state action pair $(s,a)$, 
and hence 
collectively have a sufficient impact on the reward/cost of the system. 
When this condition is not met, the cost of rare events to the users and network operator is not sufficiently high enough, resulting in less of a need to deal with rare events \textcolor{cyan}{and use importance sampling to capture their impact on the cost}. 

\subsection{Boundedness and Convergence Properties}
\label{subsection:algoProof}
In this subsection, we prove the convergence of \textbf{FIRE-ImRE}, our importance sampling based Q-learning algorithm (Algorithm \ref{algo:combined}). 

Firstly, we show that the sequence of updates in Algorithm \ref{algo:combined} is bounded, in the following theorem. 

\begin{theorem}
For stepsizes that satisfy $\sum_t \alpha(t) = \infty$ and  $\sum_t \alpha^2 (t) \leq \infty$, and discount factor $\gamma \in (0,1)$, 
the sequence of the $\hat{Q}$ updates in Algorithm \ref{algo:combined} is bounded, with probability 1.
\end{theorem}


\begin{proof}
\textcolor{teal}{See our Technical Report \cite{FIREappendix}.}
\end{proof}

The main idea of the proof is as follows: We rewrite Algorithm \ref{algo:combined}'s update equation, through adding and subtracting terms.
We rewrite the equation in the form of $\hat{Q}^{t+1}(s,a)= \hat{Q}^{t}(s,a) + \alpha^t[ F_{s,a} (\hat{Q}^t)- \hat{Q}^t(s^t,a^t)+M_s(t)]$.
Invoking Theorem 1 of \cite{tsitsiklis1994asynchronous} on stochastic approximation algorithms, we finally prove convergence by showing that the conditions for the sequence $F_{s,a} (\hat{Q})$ to be bounded are met. 
We do this by showing that $||M(t)||$ is bounded, and that the property $||F(\hat{Q}) ||_{\infty} \leq b||\hat{Q} ||_{\infty} +d$ holds, where $|| \quad ||_{\infty} $ is the supremum norm. 

Next, to help us prove that our algorithm FIRE-ImRE converges to optimality, we invoke the following corollary from \cite{jaakkola1993convergence}:
\begin{corollary}
The random process $\{ \triangle_t \}$ taking values in $\mathbb{R}^n$ and defined as $\triangle_{t+1}(x)=(1-\alpha_t(x)) \triangle_t(x) + \alpha_t F_t(x)$
converges to zero with probability 1 when the following assumptions hold:
\begin{equation*}
    0\leq \alpha_t \leq 1, \sum_t \alpha_t(x) =\infty \ \text{and} \sum_t \alpha^2_t(x) < \infty.
\end{equation*}
\begin{equation*}
    || \mathbb{E}[F_t(x)] ||\leq \gamma || \triangle_t ||, \text{where}\ \gamma <1. 
\end{equation*}
\begin{equation*}
    var[F_t(x)] \leq C(1+|| \triangle_t ||^2),\ \text{for}\ C>0.
\end{equation*}
\label{corollary:ConvergenceofRandomProcess}
\end{corollary}
In the following theorem, we show that our algorithm converges, and that it converges to the optimal policy.


\begin{theorem}
If the MDP \textcolor{black}{underlying the RL environment} is unichain for $\epsilon \in ( \delta,1-\delta)$,\footnote{$( \delta,1-\delta)$ is the range of values which $\hat{\epsilon}$ takes, as defined in Algorithm \ref{algo:combined}. If the MDP defined by transition probability $p(s'|s,a)$ is unichain for one value in $( \delta,1-\delta)$, it is unichain for all values \cite{frank2008reinforcement}.} 
for stepsizes satisfying $0 \leq \alpha_t \leq 1$, $\sum_t \alpha_i(t) = \infty$ and  $\sum_t \alpha_i^2 (t) \leq \infty$, and discount factor $\gamma \in (0,1)$,
Algorithm \ref{algo:combined} converges to the optimal policy $Q^*$.
\label{thr:optW/oEtraces}
\end{theorem}

\begin{proof}
\textcolor{teal}{See our Technical Report \cite{FIREappendix}.}
\end{proof}

The main idea of the proof is as follows.
Firstly, we show that the optimal policy of the true system 
is a fixed point of 
the following equation with our simulator's transition probability $\hat{p}(s'|s,a)$ and the importance sampling correction weight $w_t(s)$:
$Q^*(s,a)=\sum_{s' \in S} \hat{p}(s'|s,a) w_t(s) [r(s,a,s') +\gamma \max_b Q^*(s',b)].$
We define the operator $\mathbf{H}\hat{Q}=\sum_{s' \in S} \hat{p}(s'|s,a) w_t(s) [r(s,a,s') +\gamma \max_b \hat{Q}(s',b)]$,
and having shown that that $\mathbf{H} Q^*=Q^*$,
we next show that $\mathbf{H}\hat{Q}$ is a contraction mapping, through showing that $|| \mathbf{H} \hat{Q}_1- \mathbf{H} \hat{Q}_2||_{\infty} \leq \gamma ||\hat{Q}_1- \hat{Q}_2||_{\infty}$.  
Next, we invoke the result of Corollary \ref{corollary:ConvergenceofRandomProcess} on the convergence of random processes.
We rewrite the update equation of Algorithm \ref{algo:combined}, such that it fits the form of the equations 
in Corollary \ref{corollary:ConvergenceofRandomProcess}, and show that the conditions of Corollary \ref{corollary:ConvergenceofRandomProcess} hold. We do this by letting $G_t(s,a) = w^t[r(s,a, X(s,a))+\gamma \max_b \hat{Q}^t(X(s,a),b) ] - Q^*(s,a)$,
and by using the fact that $\mathbf{H}\hat{Q}$ is a contraction mapping, we show that $||\mathbb{E}[G_t(s,a)] ||_{\infty} \leq \gamma ||\hat{Q}^t(s,a)-Q^* ||_{\infty}=\gamma ||\triangle_t||_{\infty}$, and $\mathbf{var}[G_t(s,a)] \leq C (1+ || \triangle_t ||^2_{\infty})$.

\textcolor{teal}{\textbf{Rare Event Adaptive Eligibility Traces Algorithm (FIRE-ETAA):}
To speed up convergence for the tabular Q learning algorithm, we propose another variant of Algorithm \ref{algo:combined} using eligibility traces. 
This algorithm (\textbf{ETAA}) is a combination of our proposed importance sampling based Q-learning (Algorithm \ref{algo:combined}) and Watkin's $Q(\lambda)$ \cite{watkins1989learning}, which is an algorithm which combines Q-learning with eligibility traces.
We present this algorithm in our Technical Report \cite{FIREappendix}.}





\section{\textcolor{brown}{Scaling to Large State or Action Spaces}}
\label{section:FnApproxDQN}
To handle large state spaces, as might be found in real-world edge computing networks that span one or more cities and surrounding suburbs, we propose two variants of Algorithm \ref{algo:combined} using 
deep Q-learning (Algorithm FIRE-ImDQL) \textcolor{violet}{and actor critic (Algorithm FIRE-ImACRE)} respectively. 
\textcolor{brown}{The codes are available online at \cite{githubcodeFIRE}.} 

\textbf{Rare Event Adaptive Deep Q-learning algorithm (\textcolor{cyan}{FIRE-ImDQL,} Algorithm \ref{algo:DQN}):}
We modify and combine the classic deep RL algorithm in \cite{mnih2015human} with our importance sampling based rare events adaptive algorithm (Algorithm \ref{algo:combined}) by proposing an importance sampling based loss function. 
\begin{algorithm}[t]
\caption{FIRE-ImDQL: Importance Sampling for Rare Events using Deep Q Learning}\label{algo:DQN}
\begin{algorithmic}[1]
\State \textbf{Initialise:} $\hat{Q}(s,a, \boldsymbol{\theta}^{-})$ and $Q(s,a, \boldsymbol{\theta})$ randomly, $\hat{T}(s), \hat{U}(s) \gets 0$, $\hat{\epsilon}(s) \gets \frac{1}{2}.$ 
\State Select the initial state $s^0$ and action $a^0$.
\For{all timeslots $t$}
\State $\hat{\epsilon}(s^t)$ determines if an anomaly happens. Thereafter, sample according to $h(\tilde{s}|s,a)$. The new state $s^{t+1}$ and a reward value $r^{t+1}=r(s^t,a^t,s^{t+1})$ is observed. 
\State $a^{t+1} \gets \beta-\text{greedy}(s^{t+1}).$ 
\State $a^* \gets \text{argmax}_b Q(s^{t+1},b, \boldsymbol{\theta}).$ 
\If{$Q(s^{t+1},a^*, \boldsymbol{\theta})= Q(s^{t+1},a^{t+1}, \boldsymbol{\theta})$}
\State $a^* \gets a^{t+1}$
\EndIf
\State $w_t \gets \begin{cases} \epsilon(s^t)/ \hat{\epsilon}(s^t), & \text{if} \ s^{t+1} \in T.\\
    (1-\epsilon(s^t))/(1-\hat{\epsilon}(s^t)), & \text{if} \ s^{t+1} \notin T.
    \end{cases}$
\State target $y \gets w_t (r^{t+1} + \gamma \hat{Q}(s^{t+1},a^*,\boldsymbol{\theta}^{-})) $
\State $\boldsymbol{\theta}_{t+1} \gets \boldsymbol{\theta}_t + \alpha \nabla_{\boldsymbol{\theta}} \mathbb{E} [y- Q(s,a, \boldsymbol{\theta})]^2$

\State $T(s^t) \gets (1-\alpha_T) T(s^t) + \alpha_T \epsilon(s^t) (r^{t+1}+\gamma Q(s^{t+1},a^*, \boldsymbol{\theta}_{t+1}))$
\State $U(s^t) \gets (1-\alpha_U) U(s^t) + \alpha_U (1-\epsilon(s^t)) (r^{t+1}+\gamma Q(s^{t+1},a^*,\boldsymbol{\theta}_{t+1} ))$
\State $\hat{\epsilon}(s^t) \gets min(max(\delta, \frac{|T(s^t)|}{ |T(s^t)|+|U(s^t)|} ) ,1-\delta)$
\State Every K steps, reset $\hat{Q}=Q$.
\EndFor

\end{algorithmic}
\end{algorithm}
We parameterize the $\hat{Q}(s,a)$ value function using neural networks.
The loss function is 
\begin{equation*}
    L(\boldsymbol{\theta}) = \mathbb{E}_{s,a,s'} [w_t( r^{t+1}+\gamma \max_a \hat{Q}(s,a, \boldsymbol{\theta}^{-})) - Q(s,a, \boldsymbol{\theta}) ]^2,
\end{equation*}
where $\boldsymbol{\theta}^{-}$ are the weights of the target neural network,  $\boldsymbol{\theta}$ are the weights of the predicted neural network, which are updated every iteration, and $w_t$ is the importance sampling weight.
The loss function is optimized via gradient descent in line 13.
We are trying to minimize the difference between the target $w_t( r^{t+1}+\gamma \max_a \hat{Q}(s,a, \boldsymbol{\theta}^{-}))$ and the prediction $Q(s,a, \boldsymbol{\theta})$. 
The target network is updated only every $K$ iterations \textcolor{black}{(line 16)}, while the prediction network is updated every iteration. Having the target fixed for awhile helps in having stable training.
The key idea which differentiates our algorithm from traditional deep RL algorithms is that our target $w_t( r^{t+1}+\gamma \max_a \hat{Q}(s,a, \boldsymbol{\theta}^{-}))$ is ``error corrected" by the importance sampling weight $w_t$, and the rare event distribution is updated via $T, U$ and $\hat{\epsilon}$.
The variables $T$, $U$, $\hat{\epsilon}$ and $w_t$ are updated in the same manner as the tabular algorithm (Algorithm \ref{algo:combined}), through using the output of the prediction net $Q(s,a, \boldsymbol{\theta})$ where the input is state $s^{t+1}$. The actions are also selected in the same manner as Algorithm \ref{algo:combined}, through using the output of the prediction net $Q(s,a, \boldsymbol{\theta})$ where the input is state $s^t$.

\textcolor{violet}{
\begin{algorithm}[t]
\caption{FIRE-ImACRE: Importance Sampling based Advantage Actor Critic for Rare Events}\label{algo:AC}
\begin{algorithmic}[1]
\State \textbf{Initialise:} parameters $\boldsymbol{\theta}, \boldsymbol{c}$ randomly, $\hat{T}(s), \hat{U}(s) \gets 0$, $\hat{\epsilon}(s) \gets \frac{1}{2}$, learning rates $\alpha_{\boldsymbol{\theta}}, \alpha_{\boldsymbol{c}}, \alpha_T, \alpha_U$. 
\State Select the initial state $s^0$ and action $a^0$.
\For{all timeslots $t$}
\State $\hat{\epsilon}(s^t)$ determines if an anomaly happens. Thereafter, sample according to $h(\tilde{s}|s,a)$. The new state $s^{t+1}$ and a reward value $r^{t+1}=r(s^t,a^t,s^{t+1})$ is observed. 
\State Sample next action $a^{t+1} \sim \pi_{\boldsymbol{\theta}}(a|s^{t+1}) $ 
\State $w^t \gets \begin{cases} \epsilon(s^t)/ \hat{\epsilon}(s^t), & \text{if} \ s^{t+1} \in T.\\
    (1-\epsilon(s^t))/(1-\hat{\epsilon}(s^t)), & \text{if} \ s^{t+1} \notin T.
    \end{cases}$
\State $V(s^t) \gets \boldsymbol{c}(s^t)$
\State Advantage function $A(s^t,a^t) \gets r^{t+1} + \gamma V(s^{t+1}) - w^t(s^t,a^t,s^{t+1}) V(s^t)$
\State Store $A(s^t,a^t)$ in buffer.
\State $T(s^t) \gets (1-\alpha_T) T(s^t) + \alpha_T \epsilon(s^t) (r^{t+1}+\gamma V(s^{t+1}))$ 
\State $U(s^t) \gets (1-\alpha_U) U(s^t) + \alpha_U (1-\epsilon(s^t)) (r^{t+1}+\gamma V(s^{t+1}))$
\State $\hat{\epsilon}(s^t) \gets min(max(\delta, \frac{|T(s^t)|}{ |T(s^t)|+|U(s^t)|} ) ,1-\delta)$
\For{every K steps}
\State Update policy parameters using data in buffer: $\boldsymbol{\theta} \gets \boldsymbol{\theta} -\alpha_{\boldsymbol{\theta}} \nabla_{\boldsymbol{\theta}} \sum^K_{i=1} [\frac{1}{K} log \pi_{\boldsymbol{\theta}}(a^i|s^{i}) A(s^i,a^i) ] $
\State Update critic $\boldsymbol{c} \gets \boldsymbol{c} - \alpha_{\boldsymbol{c}} \nabla_{\boldsymbol{c}} \sum^K_{i=1} \frac{1}{K} [A(s^i,a^i)]^2$
\State Empty buffer.
\EndFor
\EndFor
\end{algorithmic}
\end{algorithm}}

\textcolor{violet}{\textbf{Rare Event Adaptive Actor Critic Algorithm (FIRE-ImACRE, Algorithm \ref{algo:AC})}:
Neural networks are used to approximate the ``actor'' and ``critic'' functions \cite{mnih2016asynchronous}. The critic function estimates the value (network \textbf{$c$}), based on which the actor function updates the policy (network $\mathbf{\theta}$). 
The advantage function $A(s^t,a^t)$ (line 8) measures how good taking action $a^t$ is, compared to taking the average action, at state $s^t$.
It is a function of the reward $r^{t+1}$, the values $V(s^t)$ and $V(s^{t+1})$ (obtained from the critic network $\boldsymbol{c}$ , line 7). It is corrected by the importance sampling correction weight $w^t(s^t,a^t,s^{t+1})$ (line 8), because we perform importance sampling of rare events. 
Each iteration, $A(s^t,a^t)$ will be stored in the buffer (line 9).
The variables $T, U, \hat{\epsilon}$ and $w_t$ are updated in the same manner as the tabular algorithm (Algorithm \ref{algo:combined}).
At the end of an epoch ($K$ timeslots), both the actor and critic networks are optimized, via gradient descent.
The gradients for the actor and critic networks, respectively, are: 
\begin{equation}
\nabla_{\boldsymbol{\theta}} \sum^K_{i=1} [\frac{1}{K} log \pi_{\boldsymbol{\theta}}(a^i|s^{i}) A(s^i,a^i) ].
\end{equation}
\begin{equation}
    \nabla_{\boldsymbol{c}} \sum^K_{i=1} \frac{1}{K} [A(s^i,a^i)]^2.
\end{equation}
Both gradients are functions of the advantage value, stored in the buffer every iteration (line 9).}

\textcolor{violet}{Both algorithms differ from traditional Deep Q learning and Actor Critic algorithms, as importance sampling of server failures is integrated.}

\textcolor{magenta}{\textbf{Application to other Settings with Rare Events:} Our algorithm can also be applied for other networked and wireless settings, in which there are rare events. An example application is UAV-assisted wireless networks with potential node failures. The failures are: UAV nodes failing (e.g. due to physical damage). The state space could include UAV positions and channel conditions, and the number of users served by each UAV. Possible actions include the user-UAV allocation, the bandwidth allocation, and the decision whether or not to deploying backup UAVs. 
The environment, state space, action space, and rare events would be specific to the new setting, while our importance-sampling algorithm itself would remain the same. 
}

\section{Heterogeneous Risk Tolerances}
\label{section:riskLevel}

\begin{figure*}[t]
\centering
\includegraphics[
angle=0,scale=0.37]{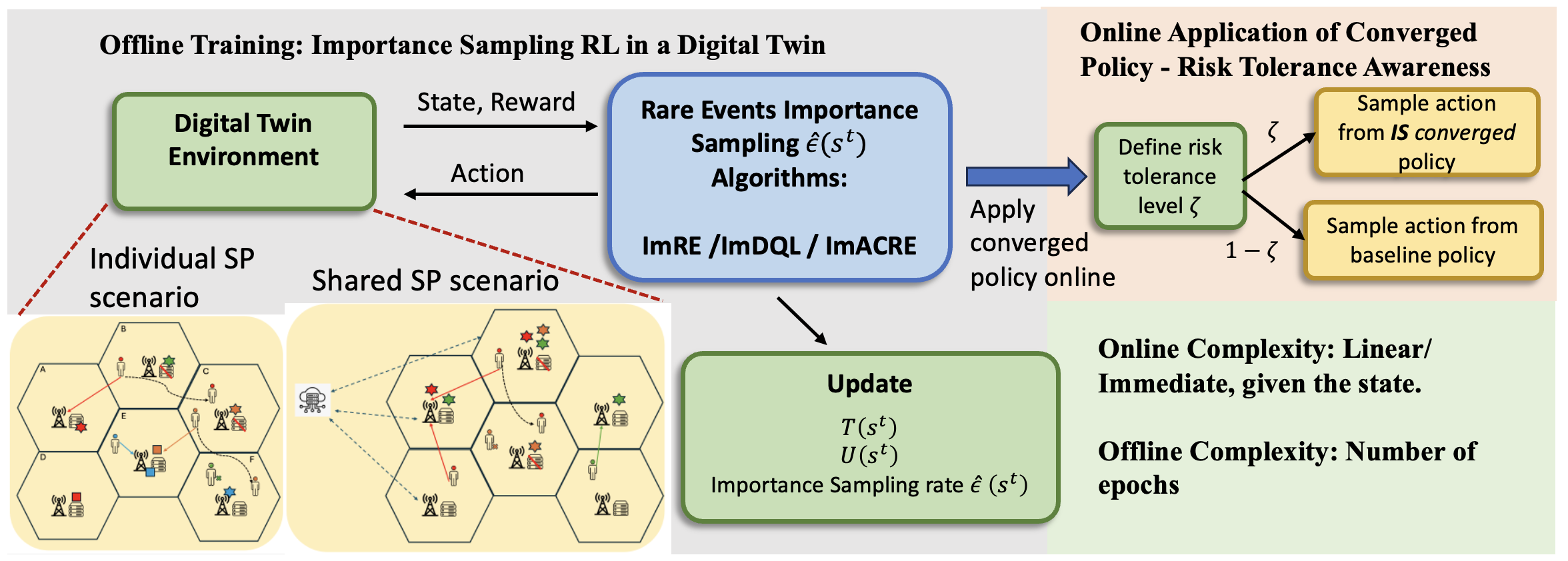}
\caption{\textcolor{magenta}{Algorithm Framework: Our importance sampling based reinforcement learning algorithms learns offline in a digital twin setup, to avoid experiencing the real cost of rare events. The converged policy is applied to online scenarios where rare events occur at their true rate.}}
\label{fig:algoFramework}
\end{figure*}


\textcolor{magenta}{We present our overall algorithm framework in Fig \ref{fig:algoFramework}. Our importance sampling based reinforcement learning algorithms learns offline in a digital twin setup, to avoid experiencing the real cost of rare events. The converged policy is applied to online scenarios where rare events occur at their true rate. Running the policy once trained is identical to deploying any RL-trained policy. If we use,
e.g., tabular RL, this is a table lookup; deep RL may require some inference time through the
neural network.}

Edge computing users, application vendors, and network providers have different tolerances towards server failures and the resulting higher latency which occur.
Some applications are highly latency sensitive, and any server failure has a significant impact on safety and smooth functioning of the application, e.g., obstacle detection in autonomous vehicles, or immersive outdoor sport applications.
Other users with less urgent jobs may not priortize a failure-aware backup placement solution, due to the extra costs of storing and migrating backups. 
\textcolor{brown}{To deal with heterogeneous risk tolerances, one potential solution is to alter the failure cost when training the algorithm. Nevertheless, the same rare event with the same system cost 
may lead to varying
risk tolerances and willingness to incur a preparation cost amongst different users.} 
Therefore, we propose a mechanism that can cater to edge computing users of heterogeneous risk tolerances. It takes the users' risk tolerance towards server failures as an input, and derives a joint service placement and backup placement policy, based on the users' risk tolerance. 
With RiTA \textcolor{orange}{(Risk Tolerance-Adaptive Algorithm)}, 
we also do not need to re-train a new policy $Q_{riskAware}$ for every possible level of risk\textcolor{cyan}{, which scales better to a wider variety of users}.

\begin{algorithm}[t]
\label{algo:riskAdaptive}
\caption{RiTA: Risk Tolerance-Adaptive Algorithm}
\begin{algorithmic}[1]
\State \textbf{Input:} User risk tolerance level $\zeta \in [0,1]$
\State $p_{riskAware} \gets \text{softmax}(Q_{riskAware}(s,a)) $
\State $p_{riskTaking} \gets \text{softmax}(Q_{riskTaking}(s,a))$
\For{all timeslots $t$}
\State State transition occurs according to $p(s'|s,a)$ and $\epsilon(s)$. 
\If{$\text{rand()} < \zeta$}
\State Sample an action according to $p_{riskTaking}$.
\Else 
\State Sample an action according to $p_{riskAware}$ 
\EndIf
\EndFor
\end{algorithmic}
\end{algorithm}

We present this algorithm in Algorithm \ref{algo:riskAdaptive} (RiTA). 
Firstly, RiTA takes the user's risk tolerance level, $\zeta$ as input. 
The lower a user's risk tolerance, the more the user wants to avoid the cost arising from server failure. 
The higher a user's risk tolerance, the less the user is willing to prepare for failure. 
Next, this algorithm converts $Q_{riskAware}$ to a vector of probabilities $p_{riskAware}$, where $Q_{riskAware}$ is the converged policy trained from one of our importance-sampling failure-adaptive Q-learning simulator algorithms (Algorithms \ref{algo:combined}, \ref{algo:DQN}, or \ref{algo:AC}).
The algorithm also converts $Q_{riskTaking}$ to a vector of probabilities $p_{riskTaking}$, where $Q_{riskTaking}$ is the converged policy of a Q-learning algorithm that does not take into account rare events in modelling of their environment, and does not consider the possibility of backups. This represents the risk taking component. 
Algorithm \ref{algo:riskAdaptive} is an online algorithm for real-time scenarios. For each timeslot, the state transition occurs according to $p(s'|s,a)$, and rare events according to the true rare events probability $\epsilon(s)$.
At each timeslot, the policy which is used will be in accordance with the risk tolerance level $\zeta \in [0,1]$.
The lower the risk tolerance $\zeta$, the higher the probability that the action selected will be sampled according to $p_{riskAware}$ \textcolor{black}{(lines 6-10)}, obtained from our algorithm's converged policy.
When $\zeta=0$, the action selected will be according to our algorithm's converged policy. 
\textcolor{violet}{In the next theorem, we characterize the condition under which it would have been beneficial to prepare for rare events.}

\textcolor{violet}{
\begin{theorem}
    After $T$ time-slots, at risk level $\zeta$, 
    assuming our algorithm is able to prevent the cost of failures with probability $1-f$, it benefited the user to prepare for rare events, if 
\begin{align*}
    & F_c > \frac{N_c ((\zeta-1)f+\zeta)}{(\zeta -(1-\zeta)f)} + \\
    & \frac{T[(1-\zeta)(N_c+B_c)-\zeta  N_c]  \sum_{i=0}^T \binom{T}{i} \epsilon^i (1-\epsilon)^{T-i} 
    }{(\zeta -(1-\zeta)f)\sum_{i=0}^T \binom{T}{i} \epsilon^i (1-\epsilon)^{T-i} i}, 
\end{align*}
where $F_c, N_c, B_c$ respectively denote the average failure cost, normal state cost, and backup placement and migration cost.
\end{theorem}
\begin{proof}
    The number of failures follow a binomial distribution with probability $\epsilon$. We use this to model the expected cost of preparation, and cost of not preparing for rare events.
\end{proof}
}

\begin{figure*}[t]
  \centering
\subfigure[Convergence of tabular importance sampling Q-learning algorithm FIRE-ImRE.]{\includegraphics[scale=0.27]{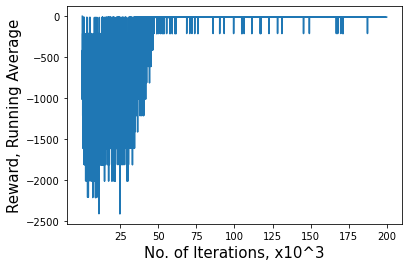}}
\hspace{0.01\textwidth}
  \subfigure[Convergence of actor critic importance sampling algorithm FIRE-ImACRE]{\includegraphics[scale=0.18]{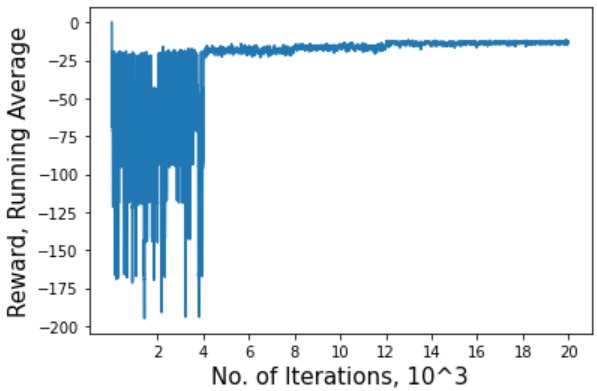}}
\hspace{0.01\textwidth}
   \subfigure[Convergence of deep Q-learning importance sampling algorithm FIRE-ImDQL, shared SP scenario.]
  {\includegraphics[scale=0.24]{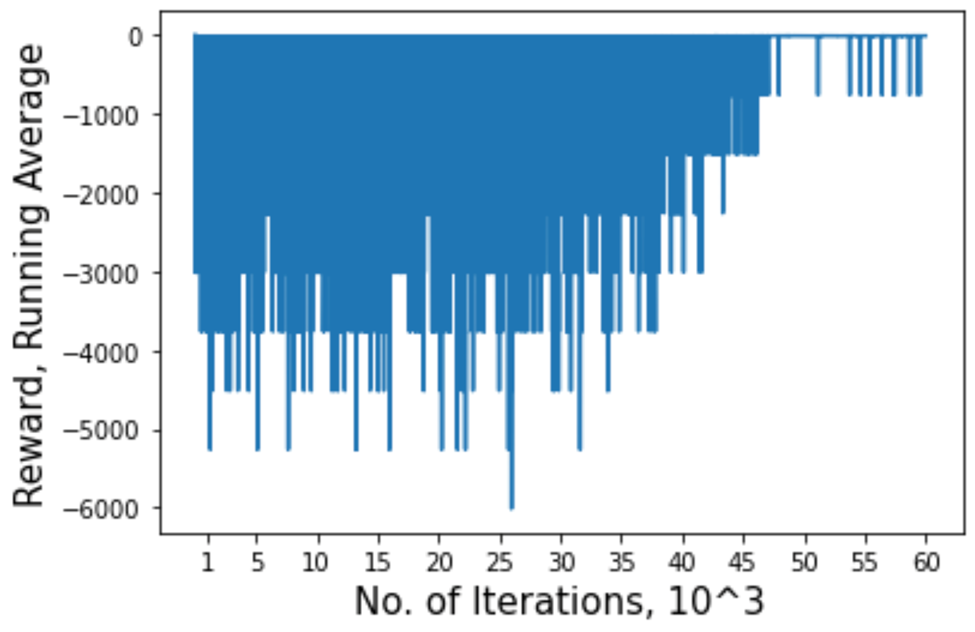}}
  \hspace{0.01\textwidth}
   \subfigure[\textcolor{magenta}{Convergence of deep Q-learning algorithm FIRE-ImDQL, individual SP scenario.}]
  {\includegraphics[scale=0.18]{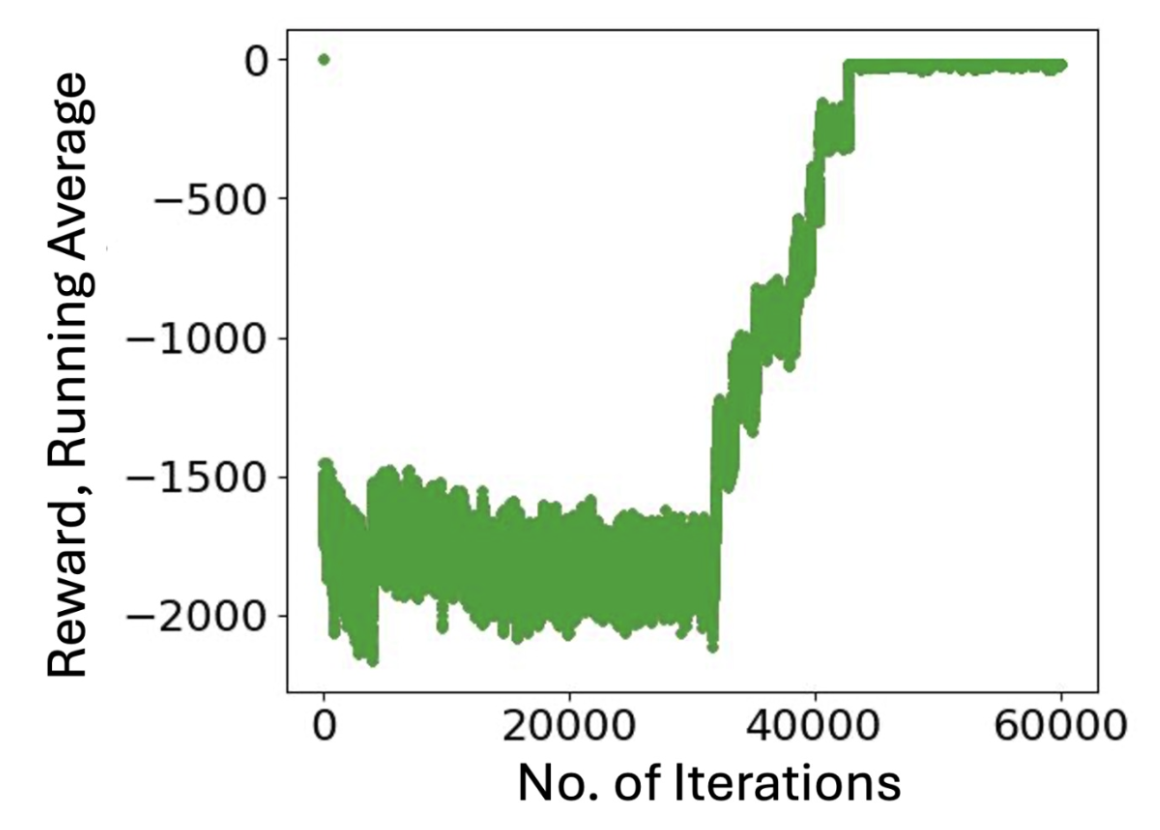}}
  \caption[Convergence.]
  {\textbf{Convergence graphs:} 
  \textcolor{orange}{Here, FIRE-ImRE and FIRE-ImACRE are applied to the a special case of the scenario (Section \ref{section:SysModel}) where every user has their own service profile (the single user case), and FIRE-ImDQL is applied to the scenario where users share service profiles (Section \ref{section:SysModel2}). \textcolor{cyan}{All three variations of our algorithm converge.}} 
  }
  \label{fig:convgGraphs}
\end{figure*}

\begin{figure*}[t]
  \centering
\subfigure[Comparison of average cost.]{\includegraphics[scale=0.2]{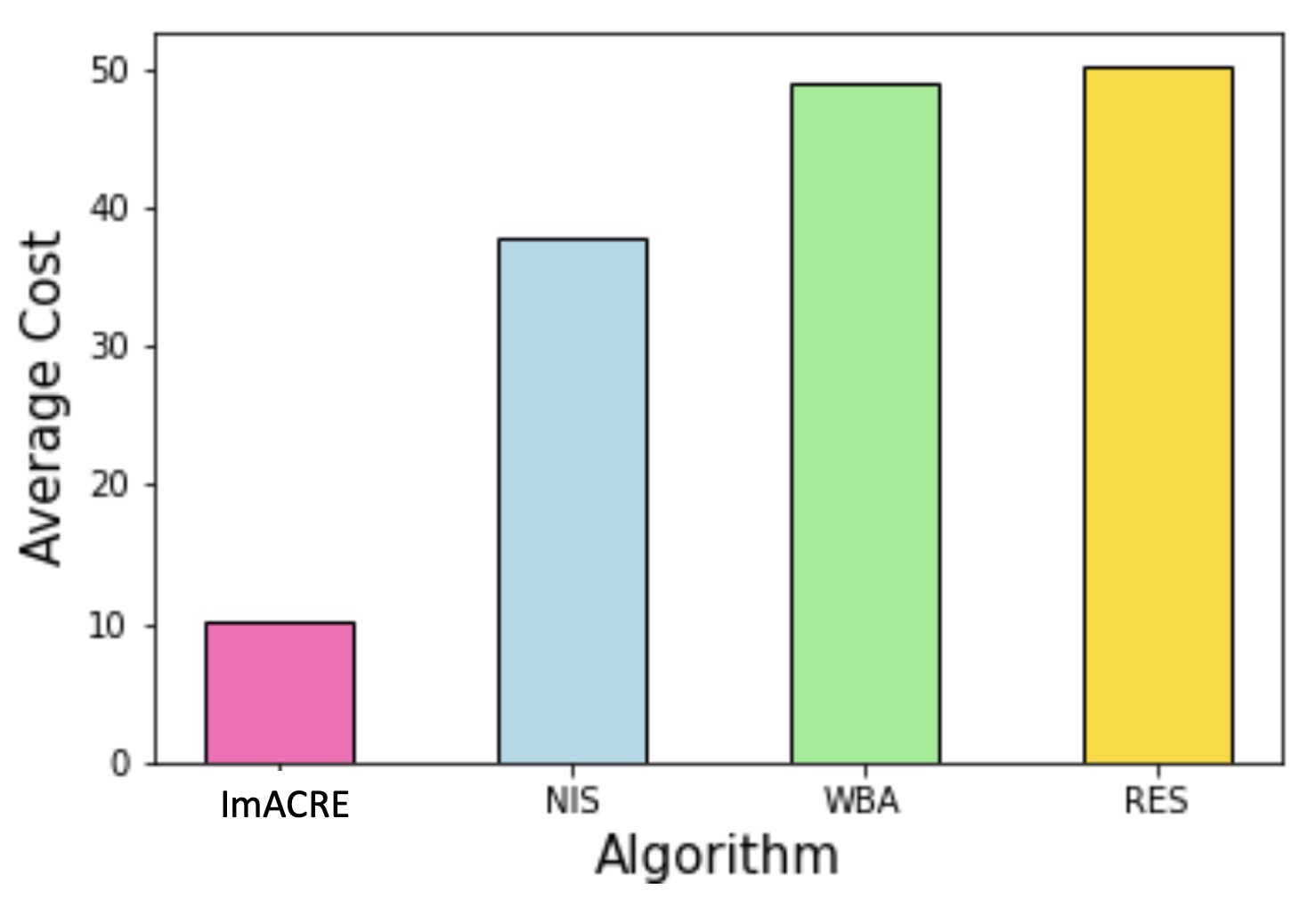}}
\hspace{0.01\textwidth}
  \subfigure[Comparison of average cost at rare state, when failure occurs.]
    {\includegraphics[scale=0.36]{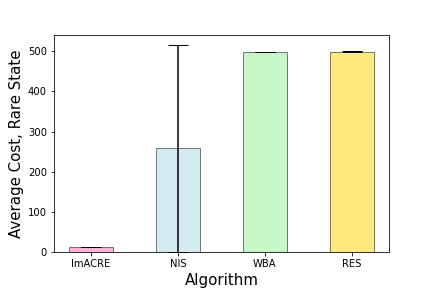}}
\hspace{0.01\textwidth}
  \subfigure[Comparison of the average cost breakdown, at normal states.]
  {\includegraphics[scale=0.36]{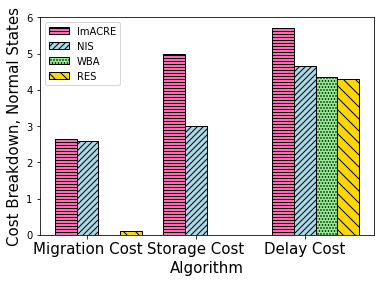}}
  \caption[Convergence.]
  {\textcolor{violet}{\textbf{Single user service migration scenario:} Comparison of our actor critic algorithm importance sampling FIRE-ImACRE, with actor critic versions of the baselines NIS, WBA and RES, in an online scenario.}\textcolor{cyan}{FIRE-ImACRE leads to lower costs on average and in rare failure states, but higher storage and delay costs in normal states.}}
  \label{fig:compAlgo_AC}
\end{figure*}

\section{Simulation \textcolor{cyan}{Experiments}}
\label{section:Simul}
In this section, we provide numerical evaluation for our algorithms FIRE-ImRE, FIRE-ImDQL, and FIRE-ImACRE, validating the results in Section \ref{section:AlgoSolution}.
\textcolor{brown}{In particular, we show that we have solved the main research challenges identified in Section \ref{section:Introduction}.
We show that 
\emph{FIRE-ImRE, FIRE-ImDQL and FIRE-ImACRE converged to optimality}.
We also show that they were able to \emph{learn a policy which mitigates the high cost of server failures}, unlike the baselines in which importance sampling or backups were not used, given a tradeoff of incurring a higher cost during normal states due to backup storage and migration.} \textcolor{cyan}{Moreover, \emph{we can adapt to varied risk tolerances without re-training.}}
We perform our experiments with the help of real world traces \cite{nsnam}, \textcolor{teal}{to obtain latency (delay) costs across different user-server location pairs in the network}.

We compare our algorithms with the following \textbf{baselines}, over 10 runs.
\textcolor{violet}{Each of these baselines, besides the greedy baseline, may be trained with deep Q-learning or actor critic (to correspond to our algorithm).} We will compare by applying the converged policy of each algorithm in an online scenario where rare events occur at their true rate.

\textit{RL with No Importance-Sampling (NIS)}: This simulator is a traditional Q-learning algorithm, without importance sampling of rare events. Rare events are simulated at their true rate, and backups are a possible action. \textcolor{cyan}{Comparison to this baseline indicates the role of importance sampling in helping the learned policy adapt to rare events.}

\textit{RL without Backups as an Action (WBA)}: This simulator is a Q-learning algorithm in which there is no importance sampling of rare events and backups are not used as possible actions. Rare events are simulated at their true rate. \textcolor{cyan}{Comparison to this baseline indicates the importance of backups in mitigating failure costs.}

\textit{RL without Rare Events Sampled (RES)}: This simulator is a Q-learning algorithm under the ``normal scenario", where rare events and backups are not considered during training. 
\textcolor{cyan}{Comparison to this baseline indicates the importance of including failures in the digital twin training environment.}


\textcolor{violet}{\textit{Greedy Placement (GPM):} Services and backups are placed at the top most likely locations of users, based on the current location of users,}
\textcolor{orange}{similar to reactive migration methods in the existing literature \cite{rejiba2019survey}.}
\textcolor{cyan}{Comparison to this baseline indicates the importance of using learned policies instead of heuristics to optimize migration policies.} 


\begin{figure*}[t]
  \centering
\subfigure[Comparison of average cost.]{\includegraphics[scale=0.34]{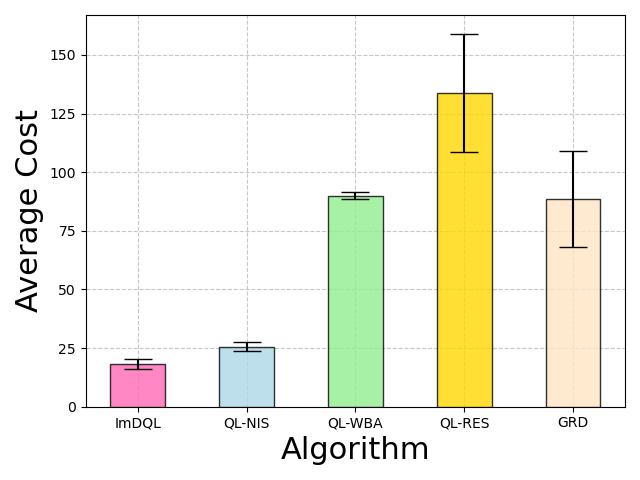}}
\hspace{0.01\textwidth}
  \subfigure[Comparison of average cost at rare state, when failure occurs.]
    {\includegraphics[scale=0.34]{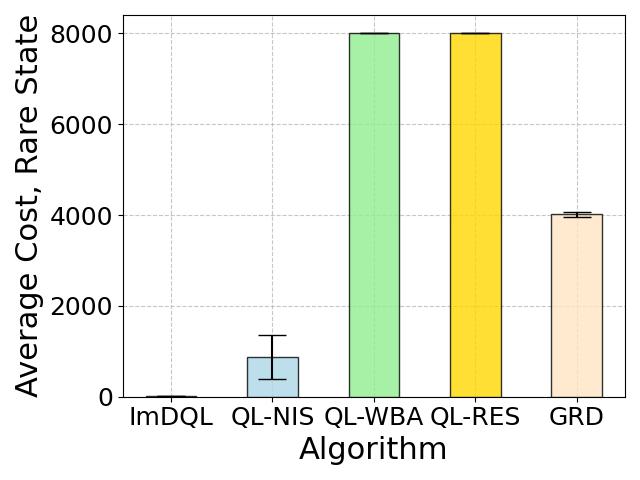}}
\hspace{0.01\textwidth}
  \subfigure[Comparison of the average cost, at normal states.]
  {\includegraphics[scale=0.34]{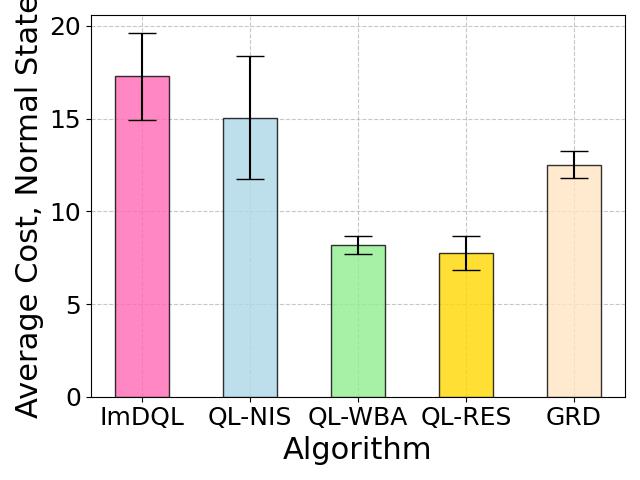}}
  \caption[Convergence.]
  {\textbf{Multiple user service profile migrations:} Comparison of our deep Q-learning algorithm FIRE-ImDQL, with baselines QL-NIS, QL-WBA and QL-RES, and greedy algorithm, in an online scenario. \textcolor{cyan}{FIRE-ImDQL yields much lower costs in rare states, yielding lower costs on average despite a slight increase in cost for normal states.}
  }
  \label{fig:MU_results}
\end{figure*}

\begin{figure*}[t]
  \centering
\subfigure[Comparison of average cost.]{\includegraphics[scale=0.39]{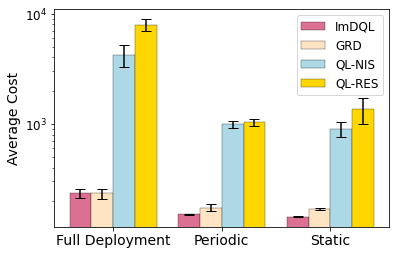}}
\hspace{0.01\textwidth}
  \subfigure[Comparison of average cost at rare state, when failure occurs.]
    {\includegraphics[scale=0.39]{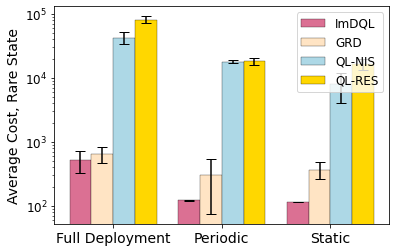}}
\hspace{0.01\textwidth}
   \subfigure[Comparison of the average cost, at normal states.]
  {\includegraphics[scale=0.39]{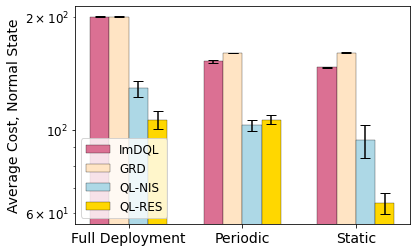}}
  \caption[Convergence.]
  {\textbf{Shared service profile migrations:} Comparison of our deep Q-learning algorithm FIRE-ImDQL, with baselines NIS, RES, and greedy algorithm, in three different settings. For all settings, our algorithm achieves a lower cost at rare states, and on average, than the baselines.
  }
  \label{fig:sharedSP_results}
\end{figure*}

\textcolor{violet}{\textbf{Single User Service Profile Migrations:} In our experiments, we have 9 base stations (edge \textcolor{magenta}{APs}). The user is highly mobile across these coverage areas, and its mobility pattern follows a transition matrix. 
\textcolor{magenta}{We use the same network settings as in \cite{kim2022modems} to simulate dynamic latencies between user locations and different base stations as users move around the service area: In particular, the ns-3 network simulator is used \cite{nsnam} to simulate dynamic channel network conditions, to obtain these latencies.} 
These latencies are in the range (2,10). As both the latency costs and migration costs are functions of the distance between user and \textcolor{magenta}{APs}, we set the migration cost $m_{ij} = d^{comm}_{i,j} + \varepsilon$, where $\varepsilon \in (-0.5,0.5)$.
We set the cost of failures to be -500, representing the high cost in terms of user experience or safety, \textcolor{orange}{and the true failure rate as} 
0.01 \textcolor{green}{(1\%)} \cite{liu2021reliability}. We choose this failure rate, because 
%
\textcolor{magenta}{if the failure rate is of a higher order (e.g. 15 percent), then we may not need a special failure-aware importance sampling strategy as the `failures' are no longer rare events - they appear often enough for normal RL methods to learn a policy which deals with them. Conversely, if the failure rate is lower (e.g. 0.001=0.1 percent), the failures would give a smaller expected cost on the system (Definition 1 in Section 3.1), so are less important to handle.}
The learning rate used is $0.05$ and both the actor and critic networks, for both our algorithms and the baselines, have a structure of $[24, 48, 24]$ \textcolor{cyan}{nodes per layer}. }

\textcolor{violet}{We apply our importance sampling based actor critic algorithm FIRE-ImACRE and the corresponding actor critic baselines NIS, WBA, RES, for this scenario.
In Fig. \ref{fig:convgGraphs}a, we show the convergence of FIRE-ImACRE.}
\textcolor{violet}{In Fig. \ref{fig:compAlgo_AC}, we plot the overall average cost, average cost at rare states, and at normal states, for our importance sampling based actor critic algorithm FIRE-ImACRE and the corresponding actor critic baselines NIS, WBA, RES. All algorithms are run 10 times, and the average is plot. Our algorithm achieves a lower average cost at rare states, because FIRE-ImACRE sufficiently samples rare events to learn a policy which prepares for them. Our algorithm incurs a trade-off with higher normal state costs, incurred through the migration and storage of backups. It can be seen in Fig. \ref{fig:compAlgo_AC}b that Algorithm NIS avoids the elevated costs associated with rare events in certain runs. However, its performance exhibits a greater variability.
WBA and RES incur the highest costs in rare events due to the lack of backup usage and the exclusion of rare events from sampling during RES simulator training.
}

\textcolor{purple}{\textbf{Multiple User Service Profile Migrations:} Here, we consider service profile migration for \textbf{multiple users}, each with their own service profile. We model 13 users mobile across 9 \textcolor{magenta}{APs}. Parameter settings are similar to the setup above. The natural failure rate is set to be 0.01. User mobility patterns are sampled according to a transition probability matrix. 
We model \textit{different server failure types}, corresponding to the differing amounts of maintenance resources at each AP, 
which affect the time taken for the servers to come back online. For APs 1,6, and 8, we let their service downtime be 2 time-slots \textcolor{cyan}{per failure}. For all other APs, we let their service downtime be 1 time-slot \textcolor{cyan}{per failure}.
We apply algorithm FIRE-ImDQL, our importance sampling based deep Q learning algorithm in this scenario. 
}
\textcolor{purple}{In Fig. \ref{fig:MU_results}, we see that our \textcolor{cyan}{FIRE-ImDQL policy} outperforms the baselines by achieving a lower average cost during rare states, with a trade-off of a higher average cost in normal (non-rare) states. Overall, our algorithm still achieves a lower cost, when averaged over all timeslots. QL-NIS and the greedy algorithm perform better than the other baselines, as backups are still a potential action for the algorithm in QL-NIS, and the greedy algorithm places one backup at the second most likely location every time-slot. Nevertheless these algorithms do not avoid the rare event cost as much as our algorithm does. }

\begin{figure*}[t]
  \centering
  \subfigure[Individual SP scenario: Average reward across rare states.]
{\includegraphics[scale=0.3]{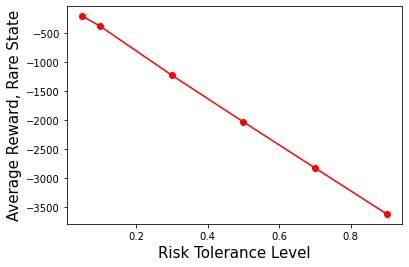}}
\hspace{0.01\textwidth}
   \subfigure[Individual SP scenario: Average reward across normal states.]
  {\includegraphics[scale=0.3]{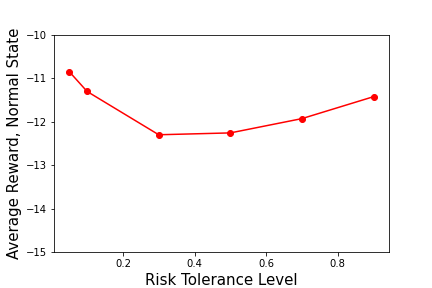}}
\hspace{0.01\textwidth}
  \subfigure[Shared SP scenario: Average Cost across all states.]{\includegraphics[scale=0.18]{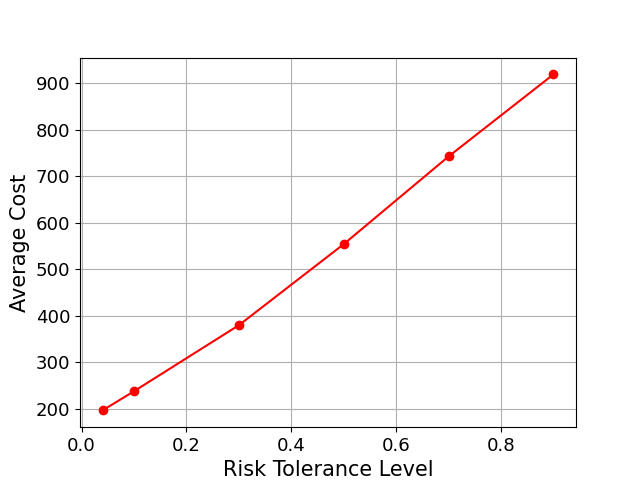}}
\hspace{0.01\textwidth}
  \subfigure[Shared SP scenario: Average Cost, normal states.]{\includegraphics[scale=0.18]{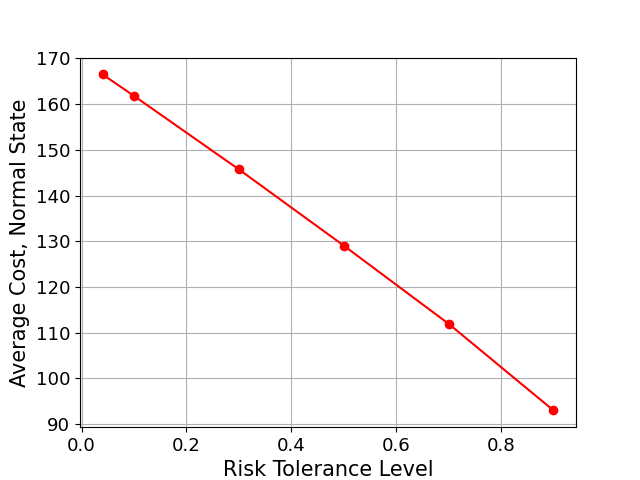}}
  \caption[Risk Level.]
  {\textbf{Heterogeneous Risk Tolerances:} Evaluating how the risk tolerance levels impact the cost, when Algorithm RiTA is applied. \textcolor{cyan}{A higher risk tolerance leads to higher cost and lower reward in rare states, but may decrease the cost in normal states.}
  }
  \label{fig:riskLevelEffect}
\end{figure*}

\textcolor{magenta}{\textbf{Comparison with PPO:} In Table \ref{tab:PPO_comp}, we present results in which our algorithm is compared with PPO (proximal policy optimization), for the multiuser scenario, where users have individual service profiles. The results are averaged over 6 runs.
It outperforms PPO when the converged policies of both algorithms are compared in an online setting where rare events occur at their true rate: PPO has a higher average cost over time, and higher rare event cost, as there is no importance sampling and hence lack of rare event data to train a policy which is adaptive to rare events.}

\begin{table}[h]
\centering
\caption{\textcolor{magenta}{Comparison of ImDQL with PPO, averaged over 6 runs. For both algorithms, the converged policy is applied in an online scenario where rare events occur at their true rate.}}
\label{tab:PPO_comp}
\begin{tabular}{|c|c|c|}
\hline
\textcolor{magenta}{\textbf{Reward across time, online setting}} & \textcolor{magenta}{\textbf{ImDQL}} & \textcolor{magenta}{\textbf{PPO}} \\
\hline
Avg reward across time & -18.210 & -22.031 \\
\hline
Mean RE reward & -18.689 & -551.79 \\
\hline
Mean RE migration cost & -6.688 & -4.628 \\
\hline
Mean RE delay & -5.176 & -5.342 \\
\hline
Mean RE storage cost & -5.000 & -4.798 \\
\hline
Mean RE computing cost & -0.9289 & -0.889 \\
\hline
Mean NS reward & -17.279 & -14.584 \\
\hline
Mean NS migration cost & -6.968 & -4.656 \\
\hline
Mean NS delay & -5.311 & -5.171 \\
\hline
Mean NS storage cost & -5.000 & -4.757 \\
\hline
Mean NS comp cost & -0.908 & -0.895 \\
\hline
\end{tabular}
\end{table}

\textcolor{magenta}{\textbf{Higher failure rates:} We have further performed new experiments with higher failure rates, namely 5 and 10\% (Tables \ref{tab:5percentComp} and \ref{tab:10percentComp}). Our algorithm is able to handle higher failure rates. Note that if failures are fairly common, then we may not need to use importance sampling to amplify rare events in the training dataset, as the rare failure events would naturally occur at a fairly high rate. Thus, reinforcement learning algorithms without importance sampling should still be able to learn good policies that account for failure costs, as these costs would be present in the training dataset. Indeed, we observe this effect in our results. At a failure rate of 5 percent (Table \ref{tab:5percentComp}), our algorithm still outperforms the baselines, achieving a much lower mean rare events reward. QL-NIS outperforms QL-WBA and QL-RES in the online scenario because backups and rare events are considered in training. At a failure rate of 10 percent (Table \ref{tab:10percentComp}), the performance of QL-NIS improves, as it is easier learn a policy in light of not-as-rare server failures. While they are of the same order, our algorithm still outperforms it.}


\begin{table}[ht]
\centering
\caption{\textcolor{magenta}{Comparison of Rewards (Costs) in Online Setting. Failure rate of 5$\%$.}}
\label{tab:5percentComp}
\begin{tabular}{|c|c|c|c|c|}
\hline
\textbf{\shortstack{Reward\\ across time}} & \textcolor{magenta}{\textbf{ImDQL}} & \textcolor{magenta}{\textbf{QL-NIS}} & \textcolor{magenta}{\textbf{QL-WBA}} & \textcolor{magenta}{\textbf{QL-RES}} \\
\hline
\shortstack{Avg reward \\ across time} & -19.586 & -26.062 & -416.322 & -578.736 \\
\hline
Mean RE reward & -20.679 & -186.516 & -8002.600 & -8003.399 \\
\hline
\shortstack{Mean RE \\ migration cost} & -7.663 & -6.284 & -2.600 & -3.399 \\
\hline
Mean RE delay & -5.189 & -4.785 & 0.000 & 0.000 \\
\hline
\shortstack{Mean RE storage \\cost} & -4.999 & -4.914 & 0.000 & 0.000 \\
\hline
\shortstack{Mean RE \\computing cost} & -0.957 & -1.029 & -0.929 & -0.993 \\
\hline
Mean NS reward & -18.394 & -15.483 & -8.061 & -7.804 \\
\hline
\shortstack{Mean NS \\ migration cost} & -8.053 & -6.332 & -2.609 & -3.456 \\
\hline
Mean NS delay & -5.344 & -4.244 & -5.452 & -4.349 \\
\hline
\shortstack{Mean NS storage \\ cost} & -4.997 & -4.907 & 0.000 & 0.000 \\
\hline
\shortstack{Mean NS\\computing cost} & -0.924 & -1.038 & -0.932 & -1.003 \\
\hline
\end{tabular}
\end{table}

\begin{table}[ht]
\centering
\caption{\textcolor{magenta}{Comparison of Rewards (Costs) in Online Setting. Failure rate of 10$\%$.}}
\label{tab:10percentComp}
\begin{tabular}{|c|c|c|c|c|}
\hline
\textbf{\shortstack{Reward \\across time}} & \textcolor{magenta}{\textbf{ImDQL}} & \textcolor{magenta}{\textbf{QL-NIS}} & \textcolor{magenta}{\textbf{QL-WBA}} & \textcolor{magenta}{\textbf{QL-RES}} \\
\hline
\shortstack{Avg reward \\across time} & -20.218 & -18.665 & -416.322 & -578.736 \\
\hline
Mean RE reward & -23.414 & -42.013 & -8002.600 & -8003.399 \\
\hline
\shortstack{Mean RE \\migration cost} & -7.643 & -5.568 & -2.600 & -3.399 \\
\hline
Mean RE delay & -5.124 & -4.652 & 0.000 & 0.000 \\
\hline
\shortstack{Mean RE \\ storage cost} & -4.997 & -4.984 & 0.000 & 0.000 \\
\hline
\shortstack{Mean RE\\ computing cost} & -0.943 & -1.027 & -0.929 & -0.993 \\
\hline
Mean NS reward & -18.265 & -14.698 & -8.061 & -7.804 \\
\hline
\shortstack{Mean NS\\ migration cost} & -7.943 & -5.633 & -2.609 & -3.456 \\
\hline
Mean NS delay & -5.328 & -4.082 & -5.452 & -4.349 \\
\hline
\shortstack{Mean NS \\storage cost} & -4.994 & -4.983 & 0.000 & 0.000 \\
\hline
\shortstack{Mean NS\\ computing cost} & -0.910 & -1.043 & -0.932 & -1.003 \\
\hline
\end{tabular}
\end{table}

\textcolor{purple}{\textbf{Shared service profiles:} In this scenario (Fig. \ref{fig:sharedSP_results}), we consider multiple users sharing the same service profile \textcolor{cyan}{(SP)}. We model 100 users, with 70 users using SP Type 1, and 30 using SP Type 2. These users are spread over a 2-city map, with 2APs in one city and 3 in another. 
It has a failure cost of 1000 and failure rate of 1e-2 \cite{liu2021reliability}. If a user's required \textcolor{magenta}{AP} is in another city, the delay cost between cities will be nearly as large as the failure cost, which is larger than the delay cost when accessing an SP within the same city. We apply FIRE-ImDQL to this scenario and compare its performance with NIS, RES and the greedy baseline. 
The greedy baseline deploys SPs to the $N$ most probable places w.r.t distribution of user transitions, where we search through all $N$ and plot the results with the lowest costs. 
}

\textcolor{purple}{We look at 3 settings: a) \textit{Full Deployment}, where users are spread across zones evenly. The greedy baseline's optimal \textcolor{cyan}{solution} is to fully deploy all SPs.
b) \textit{Periodic}, where users move across areas periodically, and c) \textit{Static}, a corner case where users stay in their current zone.
As seen in Fig. \ref{fig:sharedSP_results}, in all settings, FIRE-ImDQL achieves a lower cost in rare event states, followed by GRD \textcolor{cyan}{(greedy)}, as compared to QL-NIS and QL-RES, which place little or no backups. While there is a trade-off with a higher cost in normal states, on average over all states, FIRE-ImDQL and GRD still achieve a lower cost (Fig. \ref{fig:sharedSP_results} a) \textcolor{cyan}{than the other baselines}. \textcolor{cyan}{FIRE-ImDQL outperforms GRD in the periodic and static settings, as the greedy baseline cannot adapt to the presence of rare failures that may occur in each city.}
}

\textcolor{purple}{\textbf{Catering to varying risk tolerances:} We apply our algorithm RiTA (Algorithm \ref{algo:riskAdaptive}), which caters to users having different risk tolerances. 
Given the risk tolerance level of a user $\zeta$, with probability $1-\zeta$, RiTA will select an action according to the converged policy of our algorithm FIRE-ImRE (risk adverse component), and with probability $\zeta$ RiTA will select an action according to the converged policy of the baseline QL-RES (risk taking component).
We apply RiTA to the single user (3 APs) service migration setting, which is a special case of the individual SP scenario, and to the multiple user shared SP scenario (5 APs across 2 cities).}

\textcolor{purple}{In Fig. \ref{fig:riskLevelEffect} we present how the risk tolerance level $\zeta$ impacts the reward levels. 
As seen in Figs. \ref{fig:riskLevelEffect} a) and b), for the single user individual SP scenario, the higher the risk tolerance, the higher the cost experienced at rare states (as it is less likely backups were placed). 
For normal states, as the risk tolerance level $\zeta$ decreases from $\zeta=1$, the cost first increases due to the increase in storage of backups (resulting from being increasingly risk adverse). 
After $\zeta=0.3$, the cost decreases. 
This is because as $\zeta$ decreases, the weightage given to the risk adverse component (our algorithm FIRE-ImRE) increases. While FIRE-ImRE incurs a higher storage cost than RES, it manages to learn a policy in which frequent migrations do not occur, hence decreasing the normal state cost.
For the shared service profile scenario, as the risk tolerance $\zeta$ increases, the cost at normal state decreases due to the placement of less services, while the average cost (over all states) increases. This is because multiple users share service profiles like game environments, multiplying the cost when failures occur.}

\textcolor{magenta}{\textbf{Future work:} Scalability to a large number of users: 
If users do not share service profiles, resource allocation is mostly de-coupled, except for the `queuing delay at server' cost.
Under this setting, the algorithm may not be as scalable for larger state spaces, because in our importance sampling algorithm, we update the functions $T(s)$ and $U(s)$ at state level, necessitating us to sample each state multiple times. 
Future areas of investigation can include modifying the importance sampling-RL algorithm (which is trained offline), to sample multiple trajectories in parallel while updating $T$ and $U$, for faster convergence.
Another future area of investigation is to do “intelligent” or dynamic decoupling of the state space over base stations, via splitting the network into smaller subsets of groups of APs based on user mobility patterns. This makes the state space smaller. The algorithm can then be run on each `smaller state space'. }

\section{Conclusion}
\label{section:conclusion}
In edge computing, server failures may spontaneously occur, which complicates resource allocation and service migration. While failures are considered as rare events, they have an impact on the smooth and safe functioning of edge computing's latency sensitive applications.
As these failures occur at a low probability, it is difficult to jointly plan or learn an optimal service migration solution for both the typical and rare event scenarios. 
Therefore, we introduce a rare events adaptive resilience framework named FIRE, using the placement of backups and importance sampling based RL, which alters the sampling rate of rare events in order to learn an optimal \textcolor{cyan}{risk-averse} policy. 
We prove the boundedness and convergence to optimality of our proposed tabular Q-learning algorithm FIRE-ImRE.
To handle large and combinatorial state and action spaces common in real-world networks, we propose 
deep Q-learning (FIRE-ImDQL) and actor critic (FIRE-ImACRE) versions of our algorithm. Furthermore, we propose a decision making algorithm (RiTA) that caters to heterogeneous risk tolerances across users. 
Finally, we use trace driven experiments to show that our algorithms converge to optimality, and are resilient towards server failures, unlike several baselines, 
subject to a trade-off in terms of a potential higher normal state cost.
Our framework can also be modified for other resource allocation problems pertaining rare events with large consequences in communication and networking.








%

\bibliographystyle{IEEEtran}

\bibliography{references} 

\begin{IEEEbiographynophoto}
{Marie Siew} is a Faculty Early Career Award (FECA) Fellow at the Singapore University of Technology and Design.
\end{IEEEbiographynophoto}

\vskip -2.7\baselineskip plus -1fil
\begin{IEEEbiographynophoto}
{Shikhar Sharma} was a Masters Student in Carnegie Mellon University.
\end{IEEEbiographynophoto}
\vskip -2.7\baselineskip plus -1fil
\begin{IEEEbiographynophoto}
{Zekai Li} was a Masters Student in ECE, Carnegie Mellon University.
\end{IEEEbiographynophoto}
\vskip -2.7\baselineskip plus -1fil
\begin{IEEEbiographynophoto}
{Kun Guo} is a Professor at East China Normal University, Shanghai.
\end{IEEEbiographynophoto}
\vskip -2.7\baselineskip plus -1fil
\begin{IEEEbiographynophoto}
{Chao Xu} is a Professor in Northwest A$\&$F University, China.
\end{IEEEbiographynophoto}
\vskip -2.7\baselineskip plus -1fil
\begin{IEEEbiographynophoto}
{Tania Lorido Botran} is a Research Scientist at Roblox.
\end{IEEEbiographynophoto}
\vskip -2.7\baselineskip plus -1fil
\begin{IEEEbiographynophoto}
{Tony Q.S. Quek} is the Cheng Tsang Man Chair Professor with Singapore University of Technology and Design (SUTD) and ST Engineering Distinguished Professor.
\end{IEEEbiographynophoto}
\vskip -2.7\baselineskip plus -1fil
\begin{IEEEbiographynophoto}
{Carlee Joe-Wong} is the Robert E. Doherty Associate Professor of Electrical and Computer Engineering at Carnegie Mellon University.
\end{IEEEbiographynophoto}



\end{document}